\documentclass[11pt,onecolumn,peerreview,draftclsnofoot]{IEEEtran}
\usepackage{amsmath,amssymb,amsfonts} 
\usepackage{graphicx}                 
\usepackage{verbatim} 
\usepackage{flushend} 
\usepackage{cite}
\usepackage{algorithmic}
\long\def\symbolfootnote[#1]#2{\begingroup\def\thefootnote{\fnsymbol{footnote}}
\footnote[#1]{#2}\endgroup}
\begin{document}
\title{Canonical Dual Method for Resource Allocation and Adaptive Modulation in Uplink SC-FDMA Systems}
\author{Ayaz Ahmad, \emph{Student Member, IEEE}, and Mohamad Assaad, \emph{Member, IEEE}
\thanks{The authors are with the Department of Telecommunications, Ecole
Sup\'{e}rieure d'Electricit\'{e} (Sup\'{e}lec), $91192$
Gif-sur-Yvette, France (e-mail: \{ayaz.ahmad,
mohamad.assaad\}@supelec.fr).}}
\maketitle
\begin{abstract}
In this paper, we study resource allocation and adaptive modulation in SC-FDMA which is adopted as the multiple access scheme for the uplink in the 3GPP-LTE standard. A sum-utility maximization (SUmax), and a joint adaptive modulation and sum-cost minimization (JAMSCmin) problems are considered. Unlike OFDMA, in addition to the restriction of allocating a sub-channel to one user at most, the multiple sub-channels allocated to a user in SC-FDMA should be consecutive as well. This renders the resource allocation problem prohibitively difficult and the standard optimization tools (e.g., Lagrange dual approach widely used for OFDMA, etc.) can not help towards its optimal solution. We propose a novel optimization framework for the solution of these problems that is inspired from the recently developed canonical duality theory. We first formulate the optimization problems as binary-integer programming problems and then transform these binary-integer programming problems into continuous space canonical dual problems that are concave maximization problems. Based on the solution of the continuous space dual problems, we derive resource allocation (joint with adaptive modulation for JAMSCmin) algorithms for both the problems which have polynomial complexities. We provide conditions under which the proposed algorithms are optimal. We also propose an adaptive modulation scheme for SUmax problem. We compare the proposed algorithms with the existing algorithms in the literature to assess their performance.
\end{abstract}
\section{Introduction}
Single Carrier Frequency Division Multiple Access (SC-FDMA) is currently
attracting a lot of attention as an alternative to OFDMA in the
uplink. Its low PAPR feature has the potential to benefit the mobile terminals in term of transmit power efficiency. In fact, SC-FDMA is a single
carrier multiple access technique which utilizes single carrier
modulation and frequency domain equalization. Its overall structure and performance are similar to that of OFDMA system. Unlike the parallel transmission of the orthogonal sub-channels in OFDMA, the sub-channels are transmitted sequentially in SC-FDMA. This sequential transmission of sub-channels considerably reduces the envelope fluctuation in transmitted
waveform and results in low PAPR\cite{H.G.Myung1}. There are two types of SC-FDMA:
localized-FDMA (L-FDMA) in which the sub-channels assigned to a user
are adjacent to each other, and interleaved-FDMA (I-FDMA) in which
users are assigned with sub-channels distributed over the entire
frequency band\cite{H.G.Myung1}. In 3GPP-LTE standard\cite{3GPP.Stand1}, the current working assumption is to
use OFDMA for downlink and localized SC-FDMA for uplink.

\vspace{-5mm}
\subsection{SC-FDMA vs OFDMA from a Resource Allocation Perspective}
\vspace{-2mm}
Most of the previous work on resource allocation has focused on
power and sub-channels allocation in downlink OFDMA systems \cite{C.Y.Wong}-\cite{K.Seong}.
One of the well known approaches for solving the OFDMA resource allocation problem is exploiting its time-sharing property\cite{W.Yu}. Based on this property, it is shown in \cite{K.Seong}, and \cite{W.Yu} that for practical number of sub-channels, the resource allocation problem in OFDMA systems can be solved by Lagrange multipliers method with zero duality gap. However, none of above is directly applicable to uplink SC-FDMA. This is due to the fact that in localized SC-FDMA in addition to the restriction of allocating a sub-channel to
one user at most, the multiple sub-channels allocated to a user should be adjacent to each other as well.
Furthermore, a frequency domain equalizer is used in SC-FDMA over all the sub-channels
allocated to the user which makes the SNR expression much more
complicated than in OFDMA where the SNR on each sub-channel is
independent from the other sub-channels.

The common approach used for resource allocation in OFDMA is to formulate the mutual exclusivity restriction on sub-channels allocation as binary-integer constraint, solve the problem to get an approximated solution in continuous domain, and then discretize the continuous values into the closest binary values. But in SC-FDMA resource allocation, this approach cannot be employed. The reason is that if the problem is solved by relaxing the 0-1 constraint, then, during discretization of the continuous domain solution, the adjacency constraint on sub-channels allocation cannot be assured.
This necessitates the design
of a framework that also ensures the adjacency constraint on sub-channels allocation which is a very difficult task.
\vspace{-5mm}
\subsection{Related Work}
\vspace{-2mm}
In most of the previous work on SC-FDMA, the implementation
problems in the physical layer are studied (e.g.,\cite{R.Dinis}-\cite{3GPP.Stand2}).
The resource allocation problem in uplink SC-FDMA has also been addressed in a number of publications.
In \cite{M.Al-Rawi}, a heuristic opportunistic scheduler for allocating frequency bands to
the users in the uplink of 3G LTE systems is proposed.
In \cite{J.Lim1}, the authors have proposed a greedy sub-optimal schedular for uplink SC-FDMA systems that is based on marginal capacity maximization. In \cite{J.Lim2}, the authors revise the same framework used in \cite{J.Lim1} for developing a proportional fair scheduling scheme. However, in addition to being sub-optimal, the proposed schedulers in both \cite{J.Lim1} and \cite{J.Lim2} do not consider the sub-channels adjacency constraint which is an important physical layer requirement for localized SC-FDMA. In \cite{S.Lee}, a set of greedy sub-optimal proportional fair algorithms for localized SC-FDMA systems is proposed in the frequency-domain setting. This work respects the sub-channels adjacency constraint but does not consider any constraint on the power. In \cite{I.C.Wong}, a weighted-sum rate maximization in localized SC-FDMA systems is considered where the problem is formulated as a pure binary-integer program. Though the proposed binary-integer programming framework captures all the basic constraints of the localized SC-FDMA and allows to perform resource allocation without resorting to exhaustive search, it is still not the best solution as the 0-1 requirement turns the problem into combinatorial with exponential complexity.
Thus, keeping in view the computational complexity of the binary-integer programming, the authors have also proposed a greedy sub-optimal algorithm that is similar in spirit to the approach in \cite{J.Lim1} with an additional constraint on the adjacency of the allocated sub-channels. Moreover, all the previous work is based on rate/capacity maximization and no work to the best of our knowledge has considered power minimization in uplink SC-FDMA systems. Since the mobile terminals have limited energy, energy-economization is needed and fast power control should be considered while allocating the resources to the users in the uplink.
\vspace{-5mm}
\subsection{Motivation and Contributions}
\vspace{-2mm}
In this paper, we consider resource allocation and adaptive modulation in localized SC-FDMA systems. We consider two optimization problems: sum-utility maximization (SUmax), and joint adaptive modulation and sum-cost minimization (JAMSCmin). Both these problems are combinatorial in nature whose optimal solutions are exponentially complex in general. The performance metric considered in the SUmax problem is the total utility of the system. Utility is basically an economics concept that reflects the user satisfaction in the system. We assume that each user in the system has an associated utility function, and the objective is to
propose a polynomial-complexity resource allocation framework that could maximize the sum-utility while respecting all the constraints of localized SC-FDMA systems specific to the LTE uplink. The user utility function specific to this paper is defined as an arbitrary function that is monotonically increasing in user's SNR. The performance of the system can be further enhanced by choosing an efficient modulation scheme for each user. Therefore, based on the resource allocation, we also propose an adaptive modulation scheme, wherein an appropriate modulation is chosen for each user depending upon its effective SNR. The cost associated to each user in the JAMSCmin problem is a function that is monotonically increasing in the transmit power of that user. The objective of the JAMSCmin is to propose a low-complexity framework that jointly allocates the transmit powers, sub-channels and the modulation schemes to the users in order to minimize the total transmit power while ensuring the individual target data rates of the users as well as capturing the basic constraints of the localized SC-FDMA systems. The joint adaptive modulation in the JAMSCmin problem is important due to the fact that in order to ensure the target data rate of the users, the powers and sub-channels allocation should take into account the modulation schemes used by the users.


In this paper, we propose a novel framework for the solution of the above problems.
In our optimization framework, first we formulate the optimization problems as binary-integer programming problems. We then transform the binary-integer programming problems into canonical dual problems \cite{D.Y.Gao} in continuous space that are concave maximization problems under ceratin conditions. We provide the global optimality conditions under which
the solution to each dual canonical problem is identical to the solution
of the corresponding primal problem. We also explore some bounds on the sub-optimality of the proposed framework when the optimality conditions are not satisfied. Our proposed framework has polynomial complexity which is a significant improvement over exponential complexity.

The rest of this paper is organized as follows: Sections II provides the system model, and Section III presents the problems formulation. The canonical dual optimization framework for the solution of both the problems is provided in Section IV, and the resource allocation and adaptive modulation algorithms are derived in Section V. Section VI illustrates the numerical results, and Section VII concludes the paper.

The following notations are used in this paper. Superscripts $(.)^{T}$, and $(.)^{H}$ stand for transpose, and Hermitian of a vector or a matrix respectively. Uppercase and lowercase boldface letters denote matrices, and vectors respectively. The word ``dual" used in this paper refers to ``canonical dual".
\section{System Model}
We consider the uplink of a single cell model that utilizes
localized SC-FDMA. The generalization to multi-cell scenario
is straightforward by considering the inter cell interference in the
signal-to-interference-plus-noise ratio (SINR) expression. We make it clear that this paper does not study inter-cell interference reduction/mangement but aims to optimize the resources in each cell by
an efficient resource allocation algorithm. In the cell, $K$
users are summed to be simultaneously active. The total bandwidth $B$
is divided into $N$ sub-channels each having 12 sub-carriers. The channel is assumed to be slowly fading
or in other words assumed to exhibit block fading characteristics.
The coherence time of the channel is greater than the transmission-time-interval
(TTI) so that the channel stays relatively constant
during the TTI (in 3GPP-LTE, TTI = 0.5msec). The users' channel gains are assumed to be perfectly known.

In the following, all signals are represented by their discrete
time equivalents in the complex baseband. Assume that $N_{k}$ be the number of consecutive sub-channels allocated to user $k$ (since a sub-channel cannot be allocated to more than one user simultaneously, $\sum_{k=1}^{K}N_{k}=N$). Let $\textbf{s}_{k}=[s_{k,1}, . . ., s_{k,N_{k}}]^{T}$ be the modulated symbol vector of the $k$th user, and $\textbf{F}_{N}$ and $\textbf{F}^{H}_{N}$ denote an
$N$-point DFT and an $N$-point Inverse DFT (IDFT) matrices respectively. The assignment of the data modulated
symbols $\textbf{s}_{k}$ to the user specific set of $N_{k}$ sub-channels can
be described by a $N_{k}$-point DFT precoding matrix $\textbf{F}_{N_{k}}$, a
$N*N_{k}$ mapping matrix $\textbf{D}_{k}$ and an $N$-point IDFT matrix
$\textbf{F}^{H}_{N}$. The mapping matrix $\textbf{D}_{k}$ represents the blockwise sub-channel allocation where the elements $D_{k}(n,q)$ for $n=0, . . ., N-1$ and $q=0, . . ., N_{k}-1$ are given by
\begin{eqnarray}\label{}
   D_{k}(n,q)= \left \{ \begin{array}{ll}
     1 &
      \textrm{$n=\sum_{j=1}^{k-1}N_{j}+q$}    \\
    0 &  \textrm{elsewhere}
  \end{array} \right.
\end{eqnarray}
The transmitted signal is then
\vspace{-5mm}
\begin{equation}\label{}
    \textbf{x}_{k}=\textbf{F}_{N}^{H}\textbf{D}_{k}\textbf{F}_{N_{k}}\textbf{s}_{k}
\end{equation}
At the receiver, the received signal is transformed into the
frequency domain via a $N$-point DFT. The received signal vector for user $k$ assuming perfect sample
and symbol synchronization, is given as
\begin{equation}\label{rcvd-signal}
     \textbf{y}_{k}=\textbf{H}_{k}\textbf{F}_{N}^{H}\textbf{D}_{k}\textbf{F}_{N_{k}} \textbf{s}_{k}+\textbf{z}_{k}
\end{equation}
where $\textbf{H}_{k}=\text{diag}(h_{k,1}, . . ., h_{k,N})$ and $\textbf{z}_{k}=[z_{k,1}, . . ., z_{k,N}]^T$ are respectively the
diagonal channel response matrix and the diagonal Additive White Gaussian Noise (AWGN) vector in the frequency
domain. A frequency domain equalizer is then used in order to mitigate the ISI.
The equalized symbols are transformed back to the time
domain via an $N_{k}$-point IDFT, and the detection takes place in the time
domain.
Let $P_{k,n}$, and $\sigma^2_z$ denote the transmit power of user $k$ on sub-channel $n$, and the ambient noise variance at the receiver for user $k$ respectively.
After several manipulations, the effective SNR for user $k$ can be obtained as follows\cite{3GPP.Stand2}:
\begin{equation}
    \gamma^{ZF}_{k}=\left(\frac{1}{N_{k}}\sum_{n=1}^{N_{k}}\frac{1}{P_{k,n}G_{k,n}}\right)^{-1},\quad\quad
    \gamma^{MMSE}_{k}=\left(\frac{1}{\frac{1}{N_{k}}\sum_{n=1}^{N_{k}}\frac{P_{k,n}G_{k,n}}{1+P_{k,n}G_{k,n}}}-1\right)^{-1}
    \label{SNRs}
\end{equation}
where $\gamma^{ZF}_{k}$ is the SNR when ZF equalizer is used and $\gamma^{MMSE}_{k}$ is the SNR when MMSE equalizer is used, and where $G_{k,n}=\frac{|h_{k,n}|^{2}}{\sigma^{2}_{z}}$. The optimization framework proposed in this paper assumes an MMSE frequency domain equalization at the receiver.
Nevertheless, the proposed framework is equally applicable for ZF equalization at the receiver.

Unlike OFDMA where a different constellation can be adopted for each sub-channel, in SC-FDMA a single constellation is chosen for each user depending upon its channel quality. This is due to the fact that the transmit symbols directly modulate the sub-channels in OFDMA whereas in SC-FDMA, the transmit symbols are first fed to the FFT block and the output discrete Fourier terms are then mapped to the sub-channels. In 3GPP LTE, the constellation for each user is chosen from the set $M=\{\text{QPSK, 16QAM, 64QAM}\}$.
\section{Problems Formulation}
In this section, we formulate the two optimization problems and their equivalent binary-integer programming (BIP) problems respectively. The formulation of the problems as equivalent binary integer programs is an intermediate step towards its solution which are then approached by the canonical dual method.
\vspace{-3mm}
\subsection{Sum-Utility Maximization (SUmax)}
\subsubsection{SUmax Problem Formulation}
We want to maximize the sum-utility subject to constraint on the total transmit power of each individual
user $P^{max}_k$. We also have per sub-channel peak power constraint, $P^{peak}_{k,n}$ i.e., the peak power
transmitted on each sub-channel by any user should not exceed $P^{peak}_{k,n}$ so that the PAPR is kept low\cite{3GPP.Stand1}. In addition,
in SC-FDMA for LTE uplink, the power on all the sub-channels allocated to a user should be equal\cite{3GPP.Stand1},
so that the low PAPR benefits could retain\cite{H.G.Myung1}. The utility of user $k$ denoted as $U_k(\gamma_k)$ is an arbitrary
function that is monotonically increasing in user's SNR $\gamma_k$. The overall resource allocation problem can be
formulated as
\vspace{-2mm}
\begin{eqnarray}
\max &&\sum_{k=1}^K U_{k}(\gamma_k) \label{general}\\
\vspace{-5mm}
\text{s.t.}&&\sum_{n\in \mathcal{N}_k} P_{k,n} \leq P^{max}_k , \quad \forall k\nonumber\\
\vspace{-5mm}
&&P_{k,n} \leq P^{peak}_{k,n}, \quad \forall k, n\nonumber\\
\vspace{-5mm}
&& P_{k,n} = P_{k,l}, \quad \forall k, n, l\nonumber\\
\vspace{-5mm}
&& \mathcal{N}_{k}\cap \mathcal{N}_{j}=\emptyset, \forall k\neq j\nonumber\\
\vspace{-5mm}
&&\biggl\{ n \cap\big( \bigcup_{j=1,j\neq k}^{K}
     \mathcal{N}_{j}\big)= \emptyset \text{ }| \text{ } n \in \left\{n_{1}, n_{1}+1, . . ., n_{2}-1, n_{2}\right\}\biggr\}, \forall k  \nonumber
\vspace{-5mm}
\end{eqnarray}
where $\mathcal{N}_{k}$ with cardinality $N_k$ is the set of sub-channels allocated to users $k$, $n_{1}=\min(\mathcal{N}_{k})$ and $n_{2}=\max(\mathcal{N}_{k})$. The fourth constraint
determines that each sub-channel is allowed to be allocated to one user at most while the last constraint ensures that
the sub-channels included in the set $\mathcal{N}_{k}$ are consecutive.
The optimization problem (\ref{general}) is combinatorial in nature. There is a twofold difficulty in solving this problem, that is in addition to the exclusivity restriction on the sub-channel allocation, the allocated sub-channels to any user should be adjacent as well. For example, for $K=10$ users and $N=24$ sub-channels, the optimal solution requires a search across $5.26\times10^{12}$ possible sub-channel allocations\cite{I.C.Wong}, which is not practical.

\subsubsection{Equivalent BIP Problem for SUmax problem}
As an intermediate step towards its solution, we transform the problem to a binary-integer programming where the decisions are made on the basis of feasible set of sub-channel allocation patterns that satisfies the exclusivity and adjacency constraints and not on the basis of individual sub-channels. In other words, we form groups of contiguous sub-channels which will be optimally allocated among the users while respecting the exclusive sub-channels allocation constraint. The idea of allocation of sub-channel patterns is the same as in\cite{I.C.Wong}. We elaborate the general idea of forming the feasible sub-channel patterns with a small example. Let us suppose that we have $K=2$ users and $N=4$ sub-channels. In any allocation pattern, we put $1$ if a sub-channel is allocated to a user, and put $0$ if it is not allocated to the user. Thus, keeping in view the sub-channel adjacency constraint, the feasible set of sub-channel patterns for user $k$ can be summarized in the following matrix.
\[
 \textbf{A}^k= \left[\begin{array}{ccccccccccc}
   0 & 1 & 0 & 0 & 0 & 1 & 0 & 0 & 1 & 0 & 1 \\
   0 & 0 & 1 & 0 & 0 & 1 & 1 & 0 & 1 & 1 & 1 \\
   0 & 0 & 0 & 1 & 0 & 0 & 1 & 1 & 1 & 1 & 1 \\
   0 & 0 & 0 & 0 & 1 & 0 & 0 & 1 & 0 & 1 & 1
\end{array}\right]
\]
where each row corresponds to the sub-channel index, and each column corresponds to the feasible sub-channel allocation pattern. Note that all the $K$ users have the same allocation patterns matrix. We define a $KJ$ indicator vector $\textbf{i}=[\textbf{i}_{1}, . . ., \textbf{i}_{K}]^T$ where $\textbf{i}_k=[i_{k,1}, . . ., i_{k,J}]^T$, and where $J$ is the total number of allocation patterns. Each entry $i_{k,j}\in\{0,1\}$ which indicates whether a sub-channel pattern $j$ is allocated to a user $k$ or not.
Since a single sub-channel pattern can be allocated to each user, maximizing the users' sum-utility is equivalent to maximizing the sum-utility of all users over all sub-channel allocation patterns such that each user is assigned a single pattern while respecting the exclusive sub-channel allocation constraint. Based on this analysis we have the following lemma.
\newtheorem{lemma1}{\textbf{Lemma}}[section]
\begin{lemma1}
The sum-utility maximization problem can be written as the following binary-integer programming problem:
\vspace{-5mm}
\begin{equation}
\max_{\textbf{i}} \left\{\mathcal{P}(\textbf{i})=\sum_{k=1}^K \sum_{j=1}^Ji_{k,j}U_{k,j}(\gamma_{k,j}^{eff})\right\} \label{integer}
\vspace{-2mm}
\end{equation}
\[ \text{s.t.} \quad\sum_{k=1}^K \sum_{j=1}^Ji_{k,j}A^k_{n,j} = 1, \quad \forall n\tag{\ref{integer}a}\vspace{-5mm}\]
\[\sum_{j=1}^Ji_{k,j} = 1, \quad \forall k\tag{\ref{integer}b}\vspace{-2mm}\]
\[i_{k,j}\in\{0,1\}, \quad \forall k,j\tag{\ref{integer}c}\vspace{-5mm}\]
\end{lemma1}
where $U_{k,j}(\gamma_{k,j}^{eff})$, a monotonically increasing function of the effective SNR $\gamma^{eff}_{k,j}$ is the utility of user $k$ when allocation pattern $j$ is chosen, and $A^k_{n,j}$ denotes the element of matrix $\textbf{A}^k$ corresponding to $n$th row and $j$th column.
\begin{proof}
The proof is simple and follows from the following illustration. The effective SNR $\gamma^{eff}_{k,j}$ of user $k$ for pattern $j$ is defined as:
\vspace{-7mm}
\begin{equation}
\textstyle\gamma^{eff}_{k,j} = \left(\frac{1}{\frac{1}{N_{k,j}}\sum_{n\in \mathcal{N}_{k,j}} \frac{\min\left(P^{peak}_{k,n}, \frac{P^{max}_k}{N_{k,j}}\right)G_{k,n}}{1+\min\left(P^{peak}_{k,n}, \frac{P^{max}_k}{N_{k,j}}\right)G_{k,n}}}-1\right)^{-1}
\end{equation}
where $N_{k,j}$ is the number of sub-channels allocated to user $k$ when allocation pattern $j$ is chosen. The constraint (\ref{integer}a) ensures the exclusive sub-channel allocation i.e., any two sub-channel patterns allocated to two different users must not have any sub-channel in common. The constraint (\ref{integer}b) means that at most one allocation pattern is chosen for each user. The per-user total power, the per sub-channel peak power and the allocated sub-channels power equality constraints are all implicitly accommodated in $\gamma^{eff}_{k,j}$.
\end{proof}
\subsection{Joint Adaptive Modulation and Sum-Cost Minimization (JAMSCmin)}
\subsubsection{JAMSCmin Problem Formulation}
We now formulate the joint resource allocation and adaptive
modulation problem. The objective is to allocate powers and
sub-channels, and to choose the modulation scheme for each user in
order to minimize the sum-cost while satisfying the the target data
rate constraint of all the users (i.e., $R^{T}_k, \forall k$). For a
modulation $m\in M$ to be chosen, the effective SNR of the user
should not be less than a minimum value $\Gamma^{*}_m$ that
guarantees a target Block Error Rate (BLER) at the receiver. In
addition, the power on all the sub-channels allocated to a user
should be equal \cite{3GPP.Stand1}. In the uplink the users
terminals are more sensitive to transmit power due to their
batteries's power limitations. Therefore, we introduce in the
JAMSCmin formulation a user's cost which is function of its transmit
power and has to be minimized. We define the following cost function
for each user $k$
\vspace{-3mm}
\begin{equation}
C_k(P_k^{max},P_k) = - \exp{[P_k^{max} -P_k]}
\vspace{-4mm}
\end{equation}
where $P_k^{max}$ is the maximum power a user can transmit, and $P_k=\sum_{n\in\mathcal{N}_{k}}P_{k,n}$ is the sum of powers transmitted by user $k$ on its allocated set of sub-channels $\mathcal{N}_{k}$. The cost function is monotonically increasing in $P_k$ whereas it is monotonically decreasing in $P_k^{max}$. With this choice of cost function, the JAMSCmin problem will not only minimize the sum-power of the users but will also ensure that each user's transmit power is minimized in accordance to its $P_k^{max}$ level. In other words, a user with small $P_k^{max}$ will transmit small power compared to another user with high $P_k^{max}$, and vice versa.
The joint optimization problem can now be formulated as follows
\vspace{-3mm}
\begin{eqnarray}
\min &&\sum_{k=1}^K C_k(P_k^{max},P_k) \label{general2} \\
\vspace{-2mm}
\text{s.t.} && R_{k} \geq R^{T}_k , \forall k\nonumber\\
\vspace{-4mm}
&& P_{k,n}=P_{k,l}, \forall k,n,l\nonumber\\
&&\gamma_k \geq \Gamma^{*}_m , \forall k,m \nonumber\\
&&|\mathcal{M}_k \cap M | = 1, \forall k\nonumber\\
\vspace{-5mm}
&&\mathcal{N}_{k}\cap \mathcal{N}_{j}=\emptyset, \forall k\neq j \nonumber\\
\vspace{-5mm}
&&\biggl\{ n \cap\big( \bigcup_{j=1,j\neq k}^{K}\mathcal{N}_{j}\big)= \emptyset \text{ }| \text{ } n \in \left\{n_{1}, n_{1}+1, . . ., n_{2}-1, n_{2}\right\}\biggr\}, \forall k \nonumber
\end{eqnarray}
where $R_k$ is the $k$th user achieved data rate, $\mathcal{M}_{k}$ is a non-empty one element set that contains the modulation chosen for $k$th user; and where $\mathcal{N}_{k}$, $n_1$, and $n_2$ are the same as defined for SUmax problem. The fourth constraint reflects that a single modulation scheme is chosen for each user from the set $M$. In addition to its inherent difficulty due its  combinatorial nature as explained for the SUmax problem, the joint adaptive modulation in addition to resource allocation renders the optimization problem (\ref{general}) far more difficult to be solved.

We now formulate this joint optimization problem as an equivalent BIP problem in the following.
\subsubsection{Equivalent BIP for JAMSCmin Problem}
The sub-channel allocation patterns matrix is exactly the same as that for the SUmax problem. However, as the JAMSCmin problem considers joint adaptive modulation and resource allocation, we integrate the modulation selection into the sub-channel allocation patterns matrix.
Since the number of sub-channels needed for transmitting a certain number of bits depends on the modulation scheme used, we refine the feasible allocation pattern matrix according to the modulation schemes. For example, the minimum number of sub-channels/TTI needed for $R^T_k=140$kbps is 3, 2 and 1 for QPSK, 16QAM and 64QAM respectively. We recall that a TTI $=$ 0.5msec, and each sub-channel contains 12 sub-carriers. Thus, for example, for total number of sub-channels $N=4$, the $k$th user's feasible matrix of sub-channels allocation patterns for QPSK can be written as
\vspace{-.5mm}
\[
 \small{\textbf{B}^k_1= \left[\begin{array}{ccccccccccc}
   1 & 1 & 1 & 1 & 1 & 1 & 1 & 1 & 1 & 0 & 1 \\
   1 & 1 & 1 & 1 & 1 & 1 & 1 & 1 & 1 & 1 & 1 \\
   1 & 1 & 1 & 1 & 1 & 1 & 1 & 1 & 1 & 1 & 1 \\
   1 & 1 & 1 & 1 & 1 & 1 & 1 & 1 & 0 & 1 & 1
\end{array}\right]}
\vspace{-.5mm}\]
where the subscript $m$ in $\textbf{B}^k_m$ corresponds to the modulation index. This matrix reflects that for the given $R_k^T$, the number of sub-channels allocated to user $k$ should not be less than 3 if QPSK is chosen. The same approach can be used to define $k$th user's sub-channels allocation patterns matrices for 16QAM and 64QAM. Depending upon their target data rates, the sub-channels allocation patterns matrices can be defined for all users on all modulation schemes. We define a $KMJ$ indicator vector $\boldsymbol{\ell}=[\boldsymbol{\ell}_{1,1}, . . ., \boldsymbol{\ell}_{K,M}]^T$ where $\boldsymbol{\ell}_{k,m}=[\ell_{k,m,1}, . . ., \ell_{k,m,J}]^T$, and where $J$ is the total number of columns in the allocation pattern matrices. Each entry $\ell_{k,m,j}\in\{0,1\}$ which indicates whether a sub-channel pattern $j$ corresponding to pattern allocation matrix $\textbf{B}^k_m$ is chosen or not.
Since a single sub-channel pattern and a single modulation scheme can be chosen for each user, minimizing the users' sum-cost is equivalent to minimizing the sum-cost of all users over all sub-channel allocation pattern matrices such that each user is assigned a single pattern and a single modulation scheme while respecting the exclusive sub-channel allocation constraint.
\newtheorem{lemma2}[lemma1]{\textbf{Lemma}}
\begin{lemma2}
The joint resource allocation and adaptive modulation problem can be written as the following BIP problem:
\vspace{-7mm}
\begin{eqnarray}
\min_{\boldsymbol{\ell}}\left\{\mathrm{g}(\boldsymbol{\ell})=\sum_{k=1}^K \sum_{m=1}^M\sum_{j=1}^J\ell_{k,m,j}C_{k,j,m}(P_k^{max},P_{k,m,j})\right\} \label{integer2}
\end{eqnarray}
\vspace{-4mm}
\[\hspace{-4.5mm}\text{s.t. } \sum_{k=1}^K \sum_{m=1}^M\sum_{j=1}^J\ell_{k,m,j}B^k_{m,n,j} = 1, \quad \forall n \tag{\ref{integer2}a} \]
\vspace{-5mm}
\[\hspace{-18mm}\sum_{m=1}^M\sum_{j=1}^J\ell_{k,m,j} = 1, \forall k \tag{\ref{integer2}b} \vspace{-0.5mm}\]
\[\hspace{-18mm}\ell_{k,m,j}\in\{0,1\}, \forall k,m,j\tag{\ref{integer2}c}\vspace{-4mm}\]
\end{lemma2}
where $B^k_{m,n,j}$ denotes the element of matrix $\textbf{B}^k_{m}$ corresponding to $n$th row and $j$th column, $P_{k,m,j}=f(\gamma^{eff}_{k,m,j},R_k^T,\Gamma_m^{*})$ is the power transmitted by user $k$ when $j$th sub-channels allocation pattern corresponding to $\textbf{B}^k_m$ is chosen, and $C_{k,m,j}(P_k^{max},P_{k,m,j})=-\exp{[P_k^{max}-P_{k,m,j}]}$.
\begin{proof}
The transmit power $P_{k,m,j}$ is a function of $R_k^T$, $\Gamma_m^{*}$, and the effective SNR  $\gamma^{eff}_{k,m,j}$ of user $k$ for $j$th pattern of $\textbf{B}^k_m$. Let $P_{k,m,n}$ be the power for user $k$ on sub-channel $n$ when modulation $m$ is chosen, then $\gamma^{eff}_{k,m,j}$ is given by
\vspace{-5mm}
\begin{equation}
\gamma^{eff}_{k,m,j} = \Bigg(\frac{1}{\frac{1}{N_{k,m,j}}\sum_{n\in \mathcal{N}_{k,m,j}} \frac{P_{k,m,n}G_{k,n}}{1+P_{k,m,n}G_{k,n}}}-1\Bigg)^{-1}
\end{equation}
where $\mathcal{N}_{k,m,j}$ with cardinality $N_{k,m,j}$ is the set of sub-channels allocated to user $k$ when $j$th pattern from $\textbf{B}^k_m$ is chosen. The power allocation values $P_{k,m,j}$'s are obtained prior to resource allocation by solving the following equations:
\vspace{-5mm}
\begin{eqnarray}
\sum_{n\in \mathcal{N}_{k,m,j}}\left(\frac{P_{k,m,j}G_{k,n}}{N_{k,m,j}+P_{k,m,j}G_{k,n}}\right)-
\frac{N_{k,m,j}\Gamma_m^{*}}{1+\Gamma_m^{*}}=0, \forall k,m,j
\end{eqnarray}
which are obtained by setting $\gamma^{eff}_{k,m,j} =\Gamma_m^{*}$ and $P_{k,m,n}=\frac{P_{k,m,j}}{N_{k,m,j}}$ and hence the per user minimum SNR and the allocated sub-channels powers equality constraint are implicitly accommodated in $P_{k,m,j}$. The per-user target data rate constraint is already implicitly accommodated in the definition of allocation patterns and hence in the calculation of $P_{k,m,j}$. The constraint (\ref{integer2}a) reflects the mutual exclusivity restriction on the sub-channels allocation and constraint (\ref{integer2}b) means that at most one allocation pattern and one modulation scheme is chosen for each user.
\end{proof}
We recall that the formulation of the problems as equivalent binary-integer programs is an intermediate step towards their solution. Although the BIP problems may look simple compared to the primal problem but unfortunately, their solutions are exponentially complex due to their combinatorial nature. A similar binary-integer programming solution was proposed for weighted-sum rate maximization problem in \cite{I.C.Wong} but as mentioned before it is exponentially complex which is not practical. In the following section, we propose a polynomial-complexity framework for the solution of both the above problems that is inspired from the canonical dual transformation method.
The main idea of our proposed approach is to transform each binary-integer programming problem into a canonical dual problem in the continuous space whose solution is identical to the corresponding binary integer program under certain conditions.
\section{Canonical Dual Approach for Solving the BIP Problems}
Under certain constraints/conditions, the canonical duality theory\cite{D.Y.Gao} can be used to reformulate some non-convex/non-smooth constrained problem into certain convex/smooth canonical dual problems with perfect primal/dual relationship. However, this theory does not provide any general strategy for the solution of non-convex/non-smooth problems. The constraints under which the canonical dual problem could be perfectly dual to its primal problem is purely dependent on the nature of the primal problem under consideration and should be studied for each specific problem anew. This theory comprises of canonical dual transformation, an associated complementary-dual principle, and an associated duality theory. The canonical dual transformation can be used to convert the non-smooth problem into a smooth canonical dual problem; the complementary-dual principle can be used to study the relationship between the primal and its canonical dual problems; and the associated duality theory can help to identify both local and global extrema. Comprehensive details about this theory, and its application to an unconstrained 0-1 quadratic programming problems can be found in \cite{D.Y.Gao}, and \cite{S.C.Fang} respectively. Due to the presence of additional constraints, our problems are far more difficult compared to that described in\cite{S.C.Fang}.

By using the aforementioned theory, we transform each of the SUmax and JAMSCmin primal problems into a continuous space canonical dual problem in the following. We then study the optimality conditions, and prove that under these conditions, the solution of each canonical dual problem is identical to that of the corresponding primal problem.
\vspace{-4mm}
\subsection{Canonical Dual Problem and Optimality Conditions for SUmax Problem}
The objective function, ${\mathcal{P}}(\textbf{i})$ in problem (\ref{integer}) is a real valued linear function defined on $\mathcal{I}_a = \textbf{i} \subset \mathbb{R}^{K\times J}$ with feasible space defined by
\begin{eqnarray}
\mathcal{I}_f = \left\{\textbf{i} \in \mathcal{I}_a \subset \mathbb{R}^{K\times J} \quad | \quad \sum_{k=1}^K \sum_{j=1}^Ji_{k,j}A^k_{n,j} = 1, \forall n; \quad \sum_{j=1}^Ji_{k,j} = 1 ,\forall k;\quad i_{k,j}\in\{0,1\} \forall k,j\right\}
\label{fble-space}
\end{eqnarray}

We start our development by introducing new constraints $i_{k,j}(i_{k,j}-1)=0,\forall k,j$ which means that any $i_{k,j}$ can only take an integer value from the set $\{0,1\}$. This approach is
used for the solution of a 0-1 quadratic programming problem in \cite{S.C.Fang}. However, the problem considered in
\cite{S.C.Fang} is a simple unconstrained 0-1 quadratic programming problem while our problem is combinatorial
in nature with additional constraints. In other words, in addition to the binary-integer constraint on $i_{k,j}$'s, we
have the mutual exclusivity restriction on the allocation of sub-channel patterns $\left(\text{i.e., }\{i_{k,j}\times i_{l,j}=0 | k\neq l, \forall k,l\in\{1, . . .,K\}\}\right)$,
and the mutual exclusivity constraint on the sub-channel allocation i.e.,
$\sum_{k=1}^K \sum_{j=1}^Ji_{k,j}A^k_{n,j} = 1, \forall n$. Furthermore, at most one sub-channel pattern can be allocated to a user i.e., $\sum_{j=1}^Ji_{k,j} = 1, \forall k$. Note that the mutual exclusivity restriction on the sub-channel patterns allocation is accommodated implicitly in the formulation of the primal problem and does not show up explicitly.
We temporarily relax the new constraints $i_{k,j}(i_{k,j}-1)=0,\forall k,j$, and the equality constraints (\ref{integer}a-\ref{integer}b) to inequalities and transform the primal problem with these inequality constraints into continuous domain canonical dual problem. We will then solve the canonical dual problem in the continuous space and chose the solution which lies in $\boldsymbol{\mathcal{I}}_f$ as defined by (\ref{fble-space}).
Furthermore, for our convenience, we reformulate our primal problem as an equivalent minimization problem. The primal problem with these inequality constraints can now be written as follows.
\vspace{-2mm}
\begin{equation}
\min_{\textbf{i}} \left\{\mathrm{f}\left(\textbf{i}\right)=-\sum_{k=1}^K \sum_{j=1}^Ji_{k,j}U_{k,j}\right\} \label{relaxed}
\vspace{-3mm}
\end{equation}
\[ \text{s.t.} \quad\sum_{k=1}^K \sum_{j=1}^Ji_{k,j}A^k_{n,j} \leq 1, \quad \forall n \vspace{-3mm}\]
\[\hspace{-4mm}\sum_{j=1}^Ji_{k,j} \leq 1, \quad \forall k\vspace{-3mm}\]
\[\hspace{8mm}i_{k,j}\left(i_{k,j}-1\right)\leq 0, \quad \forall k,j \vspace{-3mm}\]
\[i_{k,j}\in\{0,1\}, \quad \forall k,j\vspace{-2mm}\]
where $U_{k,j}$ is used to denote $U_{k,j}(\gamma_{k,j}^{eff})$ and will be used in the remainder of the paper.

The temporary relaxation of the constraints to inequalities is needed for developing the canonical dual framework. We prove later that the solution of the canonical dual problem achieves the binary-integer constraints i.e., $i_{k,j}\left(i_{k,j}-1\right)=0,\forall k,j$  and all the other constraints with equality. As a first step towards its transformation into a canonical dual problem, we relax the primal problem (\cite{D.Y.Gao,S.C.Fang}). To this end, we define the so-called canonical geometrical operator $\textbf{x}=\Lambda(\textbf{i})$ for the above primal problem as follows:
\begin{eqnarray}
\textbf{x}=\Lambda(\textbf{i}) = \left( \boldsymbol{\epsilon},\boldsymbol{\lambda},\boldsymbol{\rho}\right) : \mathbb{R}^{KJ} \rightarrow \mathbb{R}^N \times\mathbb{R}^K \times \mathbb{R}^{KJ}
\end{eqnarray}
which is a vector-valued mapping and where $\boldsymbol{\rho}= [\textbf{i}^T_1(\textbf{i}_1-1), . . ., \textbf{i}^T_K(\textbf{i}_K-1)]^T$ is a KJ-vector with $\textbf{i}^T_k(\textbf{i}_k-1)=[i_{k,1}(i_{k,1}-1), . . ., i_{k,J}(i_{k,J}-1)]^T$, $\boldsymbol{\lambda}=\left[\left(\sum_{j=1}^Ji_{1,j}-1\right), . . .,\left(\sum_{j=1}^Ji_{K,j}-1\right)\right]^T$ is a K-vector and $\boldsymbol{\epsilon} = \left[\left(\sum_{k=1}^K \sum_{j=1}^Ji_{k,j}A^k_{1,j}-1\right), . . .,\left(\sum_{k=1}^K \sum_{j=1}^Ji_{k,j}A^k_{N,j}-1\right)\right]^T$ is an N-vector. Let $\chi_a$ be a convex subset of $\mathcal{\chi}=\mathbb{R}^N \times\mathbb{R}^K \times \mathbb{R}^{KJ}$ defined as follows
\begin{equation}
\chi_a = \left\{ \textbf{x}= \left(\boldsymbol{\epsilon},\boldsymbol{\lambda},\boldsymbol{\rho}\right) \in \mathbb{R}^N \times\mathbb{R}^K \times \mathbb{R}^{KJ} \quad | \quad \boldsymbol{\epsilon}\leq0,\boldsymbol{\lambda}\leq0,\boldsymbol{\rho}\leq0\right\}
\end{equation}
We introduce an indicator function $V: \chi \rightarrow \mathbb{R} \cup \{+\infty\}$, defined as
\begin{eqnarray}
V(\textbf{x}) = \left\{\begin{array} {l} 0 \quad \quad \text{  if} \quad \textbf{x} \in \chi_a, \\
+\infty \quad \text{otherwise.}
\end{array}
\right.
\end{eqnarray}
Thus, the inequality constraints in the primal problem (\ref{relaxed}) can now be relaxed by the indicator function $V(\textbf{x})$, and the primal problem can be written in the following canonical form\cite{S.C.Fang}:
\begin{eqnarray}
\min_{\textbf{i}} \left\{V(\Lambda(\textbf{i}))-\sum_{k=1}^K\sum_{j=1}^Ji_{k,j}U_{k,j} \quad | \quad i_{k,j}\in\{0,1\} \forall k,j\right\}
\label{Can-form}
\end{eqnarray}
We now define the canonical dual variables and the canonical conjugate function associated to the indicator function in order to proceed with the transformation of the primal problem into canonical dual. Since $V(\textbf{x})$ is convex, lower semi-continuous on $\chi$, the canonical dual variable $\textbf{x}^*\in\chi^*=\chi=\mathbb{R}^N \times\mathbb{R}^K \times \mathbb{R}^{KJ}$ is defined as:
\vspace{-5mm}
\begin{eqnarray}
\textbf{x}^*\in\partial V(\textbf{x}) = \left\{\begin{array} {l} \left(\boldsymbol{\epsilon}^*,\boldsymbol{\lambda}^*,\boldsymbol{\rho}^*\right) \quad \text{if} \quad \boldsymbol{\epsilon}^*\geq0 \in \mathbb{R}^N,\boldsymbol{\lambda}^*\geq0\in \mathbb{R}^K,\boldsymbol{\rho}^*\geq0\in \mathbb{R}^{KJ}, \\
\emptyset \quad \quad \quad \quad  \quad \text{  otherwise.}
\end{array}
\right.
\end{eqnarray}
By the Legendre-Fenchel transformation, the canonical super-conjugate function of $V(\textbf{x})$ is defined by
\begin{eqnarray}
V^{\sharp}(\textbf{x}^*)&=&\sup_{\textbf{x}\in\chi}\left\{\textbf{x}^T\textbf{x}^*-V(\textbf{x})\right\} =\sup_{\boldsymbol{\epsilon}\leq0}\sup_{\boldsymbol{\lambda}\leq0}\sup_{\boldsymbol{\rho}\leq0}\left\{\boldsymbol{\epsilon}^T\boldsymbol{\epsilon}^*+\boldsymbol{\lambda}^T\boldsymbol{\lambda}^*+\boldsymbol{\rho}^T\boldsymbol{\rho}^*\right\}\nonumber\\
&=& \left\{\begin{array} {l} 0 \quad \quad \text{  if} \quad \boldsymbol{\epsilon}^*\geq0,\boldsymbol{\lambda}^*\geq0,\boldsymbol{\rho}^*\geq0, \\
+\infty \quad  \text{otherwise.}
\end{array}
\right.
\end{eqnarray}
The effective domain of $V^{\sharp}(\textbf{x})$ is given by
\begin{equation}
\chi^*_a = \left\{\left(\boldsymbol{\epsilon}^*,\boldsymbol{\lambda}^*,\boldsymbol{\rho}^*\right) \in \mathbb{R}^N \times\mathbb{R}^K \times \mathbb{R}^{KJ} \quad | \quad \boldsymbol{\epsilon}^*\geq0\in \mathbb{R}^N,\boldsymbol{\lambda}^*\geq0\in \mathbb{R}^K,\boldsymbol{\rho}^*\geq0\in \mathbb{R}^{KJ}\right\}
\end{equation}
Since both $V(\textbf{x})$ and $V^{\sharp}(\textbf{x})$ are convex, lower semi-continuous, the Fenchel sup-duality relations
\begin{equation}
\textbf{x}^*\in\partial V(\textbf{x}) \Leftrightarrow \textbf{x}\in\partial V^{\sharp}(\textbf{x}^*) \Leftrightarrow
V(\textbf{x})+V^{\sharp}(\textbf{x}^*)=\textbf{x}^T\textbf{x}^*
\label{sup-duality}
\end{equation}
hold on $\chi \times \chi^*$. The pair $(\textbf{x},\textbf{x}^*)$ is called the extended / Legendre canonical dual pair on $\chi \times \chi^*$, and the functions $V(\textbf{x})$ and $V^{\sharp}(\textbf{x})$ are called canonical functions\cite{D.Y.Gao}. The optimal solution of our primal problem can be obtained if and only if $\textbf{x}=\mathcal{I}_f\in \chi_a$, i.e., along with the satisfaction of the binary-integer constraints, all the other constraints must be achieved with equality. Thus, we need to study the conditions under which the canonical dual variables $\textbf{x}^*\in \chi^*_a$ can ensure that $\textbf{x}=\mathcal{I}_f\in \chi_a$. By the definition of sub-differential, the canonical sup-duality relations (\ref{sup-duality}) are equivalent to the following:
\begin{equation}
\textbf{x}\leq 0, \quad \textbf{x}^*\geq 0, \quad \textbf{x}^T\textbf{x}^*=0
\end{equation}
From the complementarity condition $\textbf{x}^T\textbf{x}^*=0$, for $\textbf{x}^*> 0$, we have $\textbf{x} = 0$ $\left(\text{i.e.,  }\boldsymbol{\epsilon}=0,\boldsymbol{\lambda}=0,\boldsymbol{\rho}=0\right)$ and consequently $\textbf{x}=\mathcal{I}_f\in\chi_a$. This means that for $\textbf{x}^*> 0$, all the constraints of the primal problem (\ref{relaxed}) are achieved by equality (with $i_{k,j}\in\{0,1\},\forall k,j$ which comes from $\boldsymbol{\rho}=0$). Thus, the dual feasible space for the primal problem is an open positive cone defined by
\begin{equation}
\chi^*_{\sharp} = \left\{\left(\boldsymbol{\epsilon}^*,\boldsymbol{\lambda}^*,\boldsymbol{\rho}^*\right) \in \chi^*_a \quad | \quad \boldsymbol{\epsilon}^*>0,\boldsymbol{\lambda}^*>0,\boldsymbol{\rho}^*>0\right\}
\end{equation}
The so-called total complementarity function (see \cite{D.Y.Gao,S.C.Fang} for definition), $\Xi(\textbf{i},\textbf{x}^*): \chi \times \chi^*_{\sharp} \rightarrow \mathbb{R}$ associated with the primal problem (\ref{relaxed}) can be defined as follows.
\begin{equation}
\Xi(\textbf{i},\textbf{x}^*) = \Lambda(\textbf{i})^T\textbf{x}^*-V^{\sharp}(\textbf{x}^*) -\sum_{k=1}^K \sum_{j=1}^Ji_{k,j}U_{k,j}
\label{Compl-Fnc1}
\end{equation}
which is obtained by replacing $V(\Lambda(\textbf{i}))$ in (\ref{Can-form}) by $\Lambda(\textbf{i})^T\textbf{x}^*-V^{\sharp}(\textbf{x}^*)$ from Fenchel sup-duality relations\ (\ref{sup-duality}). From the definition of $\Lambda(\textbf{i})$ and $V^{\sharp}(\textbf{x}^*)$, the total complementarity function takes the form:
\begin{eqnarray}
\Xi(\textbf{i},\boldsymbol{\epsilon}^*,\boldsymbol{\lambda}^*,\boldsymbol{\rho}^*) = \sum_{k=1}^K \sum_{j=1}^J\left\{\rho^*_{k,j}i^2_{k,j}+\left( \lambda^*_k-\rho^*_{k,j}-U_{k,j}+\sum_{n=1}^N\epsilon^*_nA^k_{n,j}\right)i_{k,j}\right\}
-\sum_{n=1}^N\epsilon^*_n - \sum_{k=1}^K \lambda^*_k
\label{Compl-Fnc2}
\end{eqnarray}
Similar to \cite{S.C.Fang}, the canonical dual function $\mathrm{f}^d(\boldsymbol{\epsilon}^*,\boldsymbol{\lambda}^*,\boldsymbol{\rho}^*)$ associated to our primal problem for a given $(\boldsymbol{\epsilon}^*,\boldsymbol{\lambda}^*,\boldsymbol{\rho}^*)\in \chi^*_{\sharp}$ can be defined as
\begin{eqnarray}
\mathrm{f}^d(\boldsymbol{\epsilon}^*,\boldsymbol{\lambda}^*,\boldsymbol{\rho}^*)= \text{sta}\left\{\Xi(\textbf{i},\boldsymbol{\epsilon}^*,\boldsymbol{\lambda}^*,\boldsymbol{\rho}^*)\quad | \quad \textbf{i}\in \mathcal{I}_a\right\}
\end{eqnarray}
where $\text{sta}\{f(x)\}$ stands for finding the stationary points of $f(x)$. The complementarity function is a quadratic function of $\textbf{i}\in \mathcal{I}_a$, and has therefore a unique stationary point with respect to it for a given $(\boldsymbol{\epsilon}^*,\boldsymbol{\lambda}^*,\boldsymbol{\rho}^*)\in \chi^*_{a}$. The stationary points of $\Xi(\textbf{i},\boldsymbol{\epsilon}^*,\boldsymbol{\lambda}^*,\boldsymbol{\rho}^*)$ over $\textbf{i}\in \mathcal{I}_a$ occurs at $\textbf{i}(\textbf{x}^*)$ with
\begin{eqnarray}
i_{k,j}(\textbf{x}^*) = \frac{1}{2\rho^*_{k,j}}\left(U_{k,j} + \rho^*_{k,j} - \lambda^*_k -\sum_{n=1}^N\epsilon^*_n A^k_{n,j}\right), \quad \forall k,j
\end{eqnarray}
Replacing $i_{k,j}$ by $i_{k,j}(\textbf{x}^*)$ in (\ref{Compl-Fnc2}), we have
\begin{eqnarray}
\mathrm{f}^d(\boldsymbol{\epsilon}^*,\boldsymbol{\lambda}^*,\boldsymbol{\rho}^*)=-\frac{1}{4}\sum_{k=1}^K \sum_{j=1}^J\left\{\frac{\left(U_{k,j} + \rho^*_{k,j} - \lambda^*_k -\sum_{n=1}^N\epsilon^*_n A^k_{n,j}\right)^2}{\rho^*_{k,j}}\right\}-\sum_{n=1}^N\epsilon^*_n - \sum_{k=1}^K \lambda^*_k
\end{eqnarray}
which is a concave function in $\chi^*_{\sharp}$. The canonical dual problem associated with the primal problem (\ref{relaxed}) can now be formulated as follows
\begin{eqnarray}
\text{ext}\left\{\mathrm{f}^d(\boldsymbol{\epsilon}^*,\boldsymbol{\lambda}^*,\boldsymbol{\rho}^*) \quad | \quad (\boldsymbol{\epsilon}^*,\boldsymbol{\lambda}^*,\boldsymbol{\rho}^*)\in \chi^*_{\sharp}\right\}
\label{can-dual}
\end{eqnarray}
where the notation $\text{ext}\left\{f(x)\right\}$ stands for finding the extremum values of $f(x)$.

We have the following canonical duality theorem (Complementary-Dual Principle) on the perfect dual relationship between the primal and its corresponding canonical dual problem.
\newtheorem{theorem1}{\textbf{Theorem}}[section]
\begin{theorem1}
 If $(\overline{\boldsymbol{\epsilon}}^*,\overline{\boldsymbol{\lambda}}^*,\overline{\boldsymbol{\rho}}^*)\in\chi^*_{\sharp}$ is the stationary point of $\mathrm{f}^d(\boldsymbol{\epsilon}^*,\boldsymbol{\lambda}^*,\boldsymbol{\rho}^*) $, such that
\begin{equation}
\overline{\textbf{i}} = [\overline{i}_{1,1}, . . ., \overline{i}_{K,J}]^T \quad \text{with} \quad \overline{i}_{k,j} = \frac{1}{2\overline{\rho}^*_{k,j}}\left(U_{k,j} + \overline{\rho}^*_{k,j} - \overline{\lambda}^*_k -\sum_{n=1}^N\overline{\epsilon}^*_n A^k_{n,j}\right), \forall k,j
\label{KKT-PRIM}
\end{equation}
is the KKT point of the primal problem, and
\begin{equation}
\mathrm{f}(\overline{\textbf{i}}) = \mathrm{f}^d(\overline{\boldsymbol{\epsilon}}^*,\overline{\boldsymbol{\lambda}}^*,\overline{\boldsymbol{\rho}}^*).
\end{equation}
then the canonical dual problem (\ref{can-dual}) is perfectly dual to the primal problem (\ref{integer}).
\end{theorem1}
\begin{proof}
The proof can be obtained directly from the proof of Theorem 1 given in \cite{S.C.Fang}.
\end{proof}

The above theorem shows that the binary-integer programming problem (\ref{integer2}) is converted into a dual problem in continuous domain which is perfectly dual to it. Furthermore, the KKT point of the dual problem provides the KKT point for the primal problem. However, as the KKT conditions are necessary but not sufficient for optimality in general, we need some additional information on the global optimality. Based on the properties of the primal and dual problems, we have the following theorem on the global optimality conditions.
\newtheorem{theorem2}[theorem1]{\textbf{Theorem}}
\begin{theorem2}
If $(\overline{\boldsymbol{\epsilon}}^*,\overline{\boldsymbol{\lambda}}^*,\overline{\boldsymbol{\rho}}^*)\in\chi^*_{\sharp}$, then $\overline{\textbf{i}}$ defined by (\ref{KKT-PRIM}) is a global minimizer of $\mathrm{f}(\textbf{i})$ over $\mathcal{I}_f$ and $(\overline{\boldsymbol{\epsilon}}^*,\overline{\boldsymbol{\lambda}}^*,\overline{\boldsymbol{\rho}}^*)$ is a global maximizer of $\mathrm{f}^d(\boldsymbol{\epsilon}^*,\boldsymbol{\lambda}^*,\boldsymbol{\rho}^*)$ over $\chi^*_{\sharp}$, and
\begin{equation}
\mathrm{f}(\overline{\textbf{i}}) = \min_{\textbf{i}\in\mathcal{I}_f}\mathrm{f}(\textbf{i})
=\max_{(\boldsymbol{\epsilon}^*,\boldsymbol{\lambda}^*,\boldsymbol{\rho}^*)\in\chi^*_{\sharp}}\mathrm{f}^d(\boldsymbol{\epsilon}^*,\boldsymbol{\lambda}^*,\boldsymbol{\rho}^*)
= \mathrm{f}^d(\overline{\boldsymbol{\epsilon}}^*,\overline{\boldsymbol{\lambda}}^*,\overline{\boldsymbol{\rho}}^*).
\end{equation}
\end{theorem2}
\begin{proof}
See Appendix A.
\end{proof}
\subsection{Canonical Dual Problem and Optimality Conditions for JAMSCmin Problem}
The difference between the SUmax and the JAMSCmin problems lies only in their formulation. The procedure and steps of canonical dual transformation for JAMSCmin problem are the same as that followed for the SUmax problem, and therefore, we will not repeat them in this paper. The canonical dual problem associated to the JAMSCmin primal problem (\ref{integer2}) can be obtained as given by
\vspace{-3mm}
\begin{eqnarray}
\max\left\{\mathrm{g}^d(\boldsymbol{\xi}^*,\boldsymbol{\mu}^*,\boldsymbol{\varrho}^*) \quad | \quad (\boldsymbol{\xi}^*,\boldsymbol{\mu}^*,\boldsymbol{\varrho}^*)\in \mathcal{Y}^*_{\sharp}\right\}
\label{can-dual2}
\end{eqnarray}
where $\mathcal{Y}^*_{\sharp}$ is the associated dual feasible space defined as
\vspace{-2mm}
\begin{equation}
\mathcal{Y}^*_{\sharp} = \left\{\left(\boldsymbol{\xi}^*,\boldsymbol{\mu}^*,\boldsymbol{\varrho}^*\right) \in \mathbb{R}^N \times \mathbb{R}^K \times \mathbb{R}^{KMJ} \quad | \quad \boldsymbol{\xi}^*>0\in\mathbb{R}^N,\boldsymbol{\mu}^*>0\in\mathbb{R}^K,\boldsymbol{\varrho}^*>0\in\mathbb{R}^{KMJ}\right\}
\end{equation}
and $\mathrm{g}^d(\boldsymbol{\xi}^*,\boldsymbol{\mu}^*,\boldsymbol{\varrho}^*):
\mathbb{R}^{N}\times\mathbb{R}^{K}\times\mathbb{R}^{KMJ}\rightarrow\mathbb{R}$ is the corresponding canonical dual function defined as follows:
\vspace{-5mm}
\begin{eqnarray}
\mathrm{g}^d(\boldsymbol{\xi}^*,\boldsymbol{\mu}^*,\boldsymbol{\varrho}^*)=-\frac{1}{4}\sum_{k=1}^K\sum_{m=1}^M \sum_{j=1}^J\left\{\frac{\left(\varrho^*_{k,m,j} - C_{k,m,j}- \mu^*_k -\sum_{n=1}^N\xi^*_n B^k_{m,n,j}\right)^2}{\varrho^*_{k,m,j}}\right\}-\sum_{n=1}^N\xi^*_n - \sum_{k=1}^K \mu^*_k
\end{eqnarray}
which is a concave function on $\mathcal{Y}^*_{\sharp}$, and where $C_{k,m,j}$ is used to denote $C_{k,m,j}(P^{max}_k,P_{k,m,j})$.
Moreover, the results on the primal/dual relationship (perfect duality) and the global optimality conditions can be obtained by a similar procedure followed for SUmax problem (i.e., Theorem 4.1, and Theorem 4.2).

Based on the above mathematical analysis, we provide resource allocation (with joint adaptive modulation for JAMSCmin) algorithms in the following section. An adaptive modulation scheme for SUmax problem is also proposed since unlike the JAMSCmin problem, it does not capture the adaptive modulation implicitly in the problem formulation. The proposed adaptive modulation is based on the powers and sub-channels allocated to each user by the proposed resource allocation algorithm.
\section{Resource Allocation and Adaptive Modulation Algorithms}
\subsection{Resource Allocation Algorithm for SUmax}
The proposed algorithm is based on the solution of canonical dual problem which according to theorem 4.2 provides the optimal solution to the primal problem if the given global optimality
conditions are met. Since the dual problem is a concave maximization problem over $\chi^*_{\sharp}$, it is necessary and sufficient to solve the following system of equations for finding the optimal solution\cite{S.Boyd}.
\begin{eqnarray}
\frac{\partial\mathrm{f}^d}{\partial\epsilon^*_n}&=&\sum_{k=1}^K\sum_{j=1}^J\left\{\frac{1}{2\rho^*_{k,j}}\left(U_{k,j} + \rho^*_{k,j} - \lambda^*_k -\sum_{n=1}^N\epsilon^*_n A^k_{n,j}\right)A^k_{n,j}\right\} -1=0, \quad \forall n
\label{Eq-1}\\
\frac{\partial\mathrm{f}^d}{\partial\lambda^*_k}&=&\sum_{j=1}^J\left\{\frac{1}{2\rho^*_{k,j}}\left(U_{k,j} + \rho^*_{k,j} - \lambda^*_k -\sum_{n=1}^N\epsilon^*_n A^k_{n,j}\right)\right\} -1=0, \quad \forall k
\label{Eq-2}\\
\frac{\partial\mathrm{f}^d}{\partial\rho^*_{k,j}}&=&\left(\frac{U_{k,j} - \lambda^*_k -\sum_{n=1}^N\epsilon^*_n A^k_{n,j}}{\rho^*_{k,j}}\right)^2-1=0, \quad \forall k,j
\label{Eq-3}
\end{eqnarray}
We propose a sub-gradient based iterative algorithm for the above system of non-linear equations that is equivalent to solving $\mathrm{f}^d(\boldsymbol{\epsilon}^*,\boldsymbol{\lambda}^*,\boldsymbol{\rho}^*)$ using gradient-decent method\cite{S.Boyd}. The interest of using the sub-gradient method is its ability to use the decomposition technique that allows to simplify the solution by using a distributed method. The iterative algorithm is given in Table \ref{tab:table1} where each of $q$, $s$ and $t$ denotes the iteration number and $\beta_{\boldsymbol{\rho}^*}$, $\beta_{\boldsymbol{\lambda}^*}$ and $\beta_{\boldsymbol{\epsilon}^*}$ denote the step sizes for the sub-gradient update. For an appropriate step size, the sub-gradient method is always guaranteed to converge\cite{S.Boyd}.
The algorithm starts by initializing the variables. Then, for the given ${\boldsymbol{\epsilon}^*}^{(0)}$ and ${\boldsymbol{\lambda}^*}^{(0)}$, the solution to the set of equations (\ref{Eq-3}) i.e., ${\boldsymbol{\rho}^*}^{(q)}$ is obtained in step 1. The operation
\vspace{-5mm}
\begin{eqnarray}
&&{\boldsymbol{\rho}^*}^{(q)}\leftarrow\Pi_{\chi_{\boldsymbol{\rho}^*}}\left({\boldsymbol{\rho}^*}^{(q-1)} + \beta_{\boldsymbol{\rho}^*}\boldsymbol{\zeta}^{(q-1)}\right) :=
\nonumber \\&&
\left\{\begin{array} {l}
{\rho_{k,j}^*}^{(q)}={\rho_{k,j}^*}^{(q-1)}+ \textrm{sgn}\left({\rho_{k,j}^*}^{(q-1)}\right)\eta \quad \textrm{if} \quad ({\boldsymbol{\rho}^*}^{(q-1)} + \beta_{\boldsymbol{\rho}^*}\boldsymbol{\zeta}^{(q-1)}) = 0,\forall k,j \\
 {\rho_{k,j}^*}^{(q)}={\boldsymbol{\rho}^*}^{(q-1)} + \beta_{\boldsymbol{\rho}^*}\boldsymbol{\zeta}^{(q-1)} \quad \quad \quad \quad \quad \quad \quad \quad \quad \textrm{otherwise}.
\end{array}
\right.
\label{Proj-Eqn}
\end{eqnarray}
in step 1 is the projection of $\boldsymbol{\rho}^*$ onto the space $\chi^*_{\boldsymbol{\rho}^*}=\{{\boldsymbol{\rho}^*}\in\mathbb{R}^{KJ} | {\boldsymbol{\rho}^*} \neq 0\}$, since the canonical dual objective function is not defined at ${\boldsymbol{\rho}^*}=0$. In (\ref{Proj-Eqn}), $\textrm{sgn}$ stands for sign/signum function and $0<\eta<<1$. According to the above projection, if the updated value of ${\boldsymbol{\rho}^*}$ in the current iteration occurs to be zero, it is projected to the negative domain if its value was positive in the previous iteration, and vice versa. This projection has no impact on the convergence, since the sign of ${\boldsymbol{\rho}^*}$ does not change the direction of the gradient (see equation (\ref{Eq-3})).
Step 2 finds ${\boldsymbol{\lambda}^*}^{(s)}$ that solves equations' set (\ref{Eq-2}) for the given ${\boldsymbol{\epsilon}^*}^{(0)}$ and ${\boldsymbol{\rho}^*}^{(q)}$. These values of ${\boldsymbol{\rho}^*}^{(q)}$ and ${\boldsymbol{\lambda}^*}^{(s)}$ are then used to solve the set of equations (\ref{Eq-1}) by updating ${\boldsymbol{\epsilon}^*}^{(0)}$ to ${\boldsymbol{\epsilon}^*}^{(t)}$ in step 3. Step 4 checks whether $|\frac{\partial\mathrm{f}^d}{\partial\boldsymbol{\lambda}^*}|\leq
\delta$ for ${\boldsymbol{\rho}^*}^{(q)}$, ${\boldsymbol{\lambda}^*}^{(s)}$ and the updated
${\boldsymbol{\epsilon}^*}^{(t)}$ where $\delta\rightarrow0$ is the stopping criterion for sub-gradient update.
If $|\frac{\partial\mathrm{f}^d}{\partial\boldsymbol{\lambda}^*}|> \delta$, steps 2 through 4 are repeated until
both $|\frac{\partial\mathrm{f}^d}{\partial\boldsymbol{\lambda}^*}|\leq\delta$ and
$|\frac{\partial\mathrm{f}^d}{\partial\boldsymbol{\epsilon}^*}|\leq \delta$. In step 6, $\boldsymbol{\zeta}^{(q)}$
is recomputed for ${\boldsymbol{\rho}^*}^{(q)}$, and the updated ${\boldsymbol{\lambda}^*}^{(s)}$ and
${\boldsymbol{\epsilon}^*}^{(t)}$. If $|\boldsymbol{\zeta}^{(q)}|\leq\delta$, the algorithm is stopped otherwise
steps 1 through 6 are repeated until convergence. The resource allocation vector $\overline{\textbf{i}}$ is
then obtained from the dual optimal solution
$(\overline{\boldsymbol{\epsilon}}^*,\overline{\boldsymbol{\lambda}}^*,\overline{\boldsymbol{\rho}}^*)$ in step 8.
\begin{table}[h]
\caption{Resource Allocation Algorithm}
\centering {
\begin{tabular}{ l }
\hline
\textbf{Initialize} $({\boldsymbol{\epsilon}^*}^{(0)},{\boldsymbol{\lambda}^*}^{(0)},{\boldsymbol{\rho}^*}^{(0)}) \in\chi^*_{\sharp}$\\
\textbf{1. Compute} $\boldsymbol{\zeta}^{(q)}={\frac{\partial\mathrm{f}^d}{\partial\boldsymbol{\rho}^*}}|_{{\boldsymbol{\rho}^*}^{(q)}}$. If $|\boldsymbol{\zeta}^{(q)}|\leq\delta$, go to step 2.\\
$\quad\quad\quad\bullet$ Set ${\boldsymbol{\rho}^*}^{(q+1)}\leftarrow\Pi_{\chi_{\boldsymbol{\rho}^*}}\left({\boldsymbol{\rho}^*}^{(q)} + \beta_{\boldsymbol{\rho}^*}\boldsymbol{\zeta}^{(q)}\right)$.\\
$\quad\quad\quad \bullet$ Set $q\leftarrow q+1$, and repeat step 1.\\
\textbf{2. Compute} $\boldsymbol{\eta}^{(s)}=\frac{\partial\mathrm{f}^d}{\partial\boldsymbol{\lambda}^*}|_{{\boldsymbol{\lambda}^*}^{(s)}}$. If $|\boldsymbol{\eta}^{(s)}|\leq\delta$, go to step 3.\\
$\quad\quad\quad \bullet$ Set ${\boldsymbol{\lambda}^*}^{(s+1)}\leftarrow\left({\boldsymbol{\lambda}^*}^{(s)} + \beta_{\boldsymbol{\lambda}^*}\boldsymbol{\eta}^{(s)}\right)$.\\
$\quad\quad\quad \bullet$ Set $s\leftarrow s+1$, and repeat step 2.\\
\textbf{3. Compute} $\boldsymbol{\upsilon}^{(t)}=\frac{\partial\mathrm{f}^d}{\partial\boldsymbol{\epsilon}^*}|_{{\boldsymbol{\epsilon}^*}^{(t)}}$. If $|\boldsymbol{\upsilon}^{(t)}|\leq\delta$, go to step 4.\\
$\quad\quad\quad \bullet$ Set ${\boldsymbol{\epsilon}^*}^{(t+1)} \leftarrow \left({\boldsymbol{\epsilon}^*}^{(t)} +\beta_{\boldsymbol{\epsilon}^*}\boldsymbol{\upsilon}^{(t)}\right)$.\\
$\quad\quad\quad\bullet$ Set $t\leftarrow t+1$, and repeat step 3.\\
\textbf{4. Recompute} $\boldsymbol{\eta}^{(s)}=\frac{\partial\mathrm{f}^d}{\partial\boldsymbol{\lambda}^*}|_{{\boldsymbol{\lambda}^*}^{(s)}}$.\\
\textbf{5. Repeat steps 2 through 4 until} $|\boldsymbol{\eta}^{(s)}|\leq\delta$, and $|\boldsymbol{\upsilon}^{(t)}|\leq\delta$\\
\textbf{6. Recompute} $\boldsymbol{\zeta}^{(q)}={\frac{\partial\mathrm{f}^d}{\partial\boldsymbol{\rho}^*}}|_{{\boldsymbol{\rho}^*}^{(q)}}$\\
\textbf{7. Repeat steps 1 through 6 until} $|\boldsymbol{\zeta}^{(q)}|\leq\delta$, $|\boldsymbol{\eta}^{(s)}|\leq\delta$, and $|\boldsymbol{\upsilon}^{(t)}|\leq\delta$\\
\textbf{8. Compute} $\overline{\textbf{i}}$ \textbf{according to (\ref{KKT-PRIM})}.\\
\hline
\end{tabular}
}
\label{tab:table1}
\end{table}

\subsubsection{Adaptive Modulation Scheme for SUmax} By knowing
perfectly the effective SNR of each user from the powers and
sub-channels allocation performed according to the previous
subsection, we propose an adaptive modulation scheme in this
subsection. The proposed adaptive modulation scheme is based on the
criterion of target Target Block Error Rate (BLER) at the receiver
used for the JAMSCmin problem. According to this approach, for a
modulation $m\in M$ to be chosen, the effective SNR of the user
should not be less than a minimum value $\Gamma^*_m$ that guarantees
a target BLER at the receiver. Since the effective SNR of users are
perfectly known from the the powers and sub-channels allocation
performed according to the previous subsection, we adopt the
modulation for each user which maximizes its individual utility.
Thus, depending upon $\gamma^{eff}_k$, the efficient modulation for
user $k$ is determined as follows:
\begin{eqnarray}
m^*(k) = \arg \min_{m \in M} \left\{(\gamma^{eff}_k - \Gamma^*_m)|_{\Gamma^*_m\leq \gamma^{eff}_k}\right\}
\end{eqnarray}
Note that the above approach is similar in spirit to the approach used in \cite{Y.Liu} where adaptive modulation in OFDM system is considered and an efficient constellation is chosen for each sub-channel.
\subsection{Joint Adaptive Modulation and Resource Allocation Algorithm for JAMSCmin}
The dual function $\mathrm{g}^d(\boldsymbol{\xi}^*,\boldsymbol{\mu}^*,\boldsymbol{\varrho}^*)$  is a concave function over $(\boldsymbol{\xi}^*,\boldsymbol{\mu}^*,\boldsymbol{\varrho}^*)\in \mathcal{Y}^*_{\sharp}$. Thus, the corresponding dual problem is a concave maximization problem over $\mathcal{Y}^*_{\sharp}$ where the joint adaptive modulation and resource allocation can be obtained by solving the following set of equations:
\begin{eqnarray}
\frac{\partial\mathrm{g}^d}{\partial\xi^*_n}&=&\sum_{k=1}^K\sum_{m=1}^M\sum_{j=1}^J\left\{\frac{1}{2\varrho^*_{k,m,j}}\left( \varrho^*_{k,m,j} -C_{k,m,j}  - \mu^*_k -\sum_{n=1}^N\xi^*_n B^k_{m,n,j}\right)B^k_{m,n,j}\right\} -1=0, \forall n
\label{Eq-12}\\
\frac{\partial\mathrm{g}^d}{\partial\mu^*_k}&=&\sum_{m=1}^M\sum_{j=1}^J\left\{\frac{1}{2\varrho^*_{k,j}}\left( \varrho^*_{k,m,j} -C_{k,m,j}- \mu^*_k -\sum_{n=1}^N\xi^*_n B^k_{m,n,j}\right)\right\} -1=0, \quad \forall k
\label{Eq-22}\\
\frac{\partial\mathrm{g}^d}{\partial\varrho^*_{k,m,j}}&=&\left(\frac{-C_{k,m,j} - \mu^*_k-\sum_{n=1}^N\epsilon^*_n B^k_{m,n,j}}{\varrho^*_{k,m,j}}\right)^2-1=0, \quad \forall k,m,j
\label{Eq-32}
\end{eqnarray}

A similar procedure of sub-gradient is proposed where an iterative algorithm can be derived that is similar in spirit to that derived for the SUmax problem. Since it uses a similar procedure and has a similar sequence of steps as that for the algorithm given in Table I, the latter can be adopted to the JAMSCmin problem, and we do not reproduce it in this paper.
\vspace{-2mm}
\subsection{Complexity of the algorithm}
\subsubsection{Complexity of the algorithm for SUmax problem}In each iteration for
$\boldsymbol{\rho}^*$, we compute $KJ$ variables. The number of
variables computed in each iteration for $\boldsymbol{\lambda}^*$ is
$K$ and that for $\boldsymbol{\epsilon}^*$ is $N$. Assume that the
number of iterations required for optimal $\boldsymbol{\rho}^*$,
$\boldsymbol{\lambda}^*$ and $\boldsymbol{\epsilon}^*$ are
$I_{\boldsymbol{\rho}^*}$, $I_{\boldsymbol{\lambda}^*}$ and
$I_{\boldsymbol{\epsilon}^*}$ respectively, then the algorithm has
an overall complexity of
$\mathcal{O}(I_{\boldsymbol{\rho}^*}KJ+I_{\boldsymbol{\lambda}^*}K+I_{\boldsymbol{\epsilon}^*}N)$.
\vspace{-1mm}
\subsubsection{Complexity of the algorithm for JAMSCmin problem}
The complexity of the proposed algorithm adopted to the JAMSCmin
problem is $\mathcal{O}(I_{\boldsymbol{\varrho}^*}KMJ+I_{\boldsymbol{\mu}^*}K+I_{\boldsymbol{\xi}^*}N)$ where $I_{\boldsymbol{\varrho}^*}$, $I_{\boldsymbol{\mu}^*}$ and $I_{\boldsymbol{\xi}^*}$ are the numbers of iterations needed for finding the optimal values of $KMJ$ variables $\boldsymbol{\varrho}^*$, $K$ variables $\boldsymbol{\mu}^*$ and the $N$ variables $\boldsymbol{\xi}^*$ respectively.
\subsection{On the Optimality of the Algorithm}
The canonical dual problem is a concave maximization problem over
$\chi^*_{\sharp}$, the proposed algorithm is then surely optimal if
$(\overline{\boldsymbol{\epsilon}}^*,\overline{\boldsymbol{\lambda}}^*,\overline{\boldsymbol{\rho}}^*)
\in\chi^*_{\sharp}$. However, if
$(\overline{\boldsymbol{\epsilon}}^*,\overline{\boldsymbol{\lambda}}^*,\overline{\boldsymbol{\rho}}^*)$
is not inside the positive cone $\chi^*_{\sharp}$, then the
canonical problem is not guaranteed to be concave. Consequently, the
proposed algorithm may not find the optimal solution. From
our simulation results, we have observed that for moderate number of
sub-channels the proposed algorithm works well, and the canonical
dual solution is very close to the optimal solution.

In this subsection, we analyze the gap between the optimal solution and the solution obtained by using our proposed sub-gradient based algorithm. We perform the analysis for SUmax problem which is equally applicable to the JAMSCmin problem, and we will not repeat it in this paper. We
start the analysis by introducing a modified problem whose optimal solution is not necessary and will not replace our actual problem but is used only to study the optimality gap of our proposed algorithm. In our analysis, first we find the solution of the modified problem (which is a stationary point and may not be necessarily the optimal solution of this modified problem). Then, we show in Theorem \ref{Opt-Theo1} that there exist a primal problem with a slightly different values of the utilities $U_{k,j}$'s whose optimal solution is equal to the solution of this modified problem. Finally, in Corollary \ref{Corollary} we show that under certain conditions, the solution of the canonical dual problem obtained using the algorithm in Table \ref{tab:table1} provides solution to the primal problem which is very close to optimal solution.
%
Let us consider the following modified problem
\vspace{-4mm}
\begin{eqnarray}
(\mathcal{P}1):\text{ }\max_{\boldsymbol{\epsilon}^*,\boldsymbol{\lambda}^*,\boldsymbol{\rho}^*}\mathrm{f}^d(\boldsymbol{\epsilon}^*,\boldsymbol{\lambda}^*,\boldsymbol{\rho}^*)
\label{New-Conc-Prob}
\end{eqnarray}
\[\hspace{-7mm}\text{s.t.}\quad\boldsymbol{\epsilon}^* \geq \mathbf{c}\tag{\ref{New-Conc-Prob}a}\vspace{-6mm}\]
\[\quad\boldsymbol{\lambda}^* \geq \mathbf{d} \tag{\ref{New-Conc-Prob}b}\vspace{-4mm}\]
where $(\mathbf{c},\mathbf{d})\in (\mathbb{R}^N_+,\mathbb{R}^K_+)$.
We solve this problem using the standard Lagrangian technique.
Let $(\boldsymbol{\epsilon}^*,\boldsymbol{\lambda}^*,\boldsymbol{\rho}^*)$
be the obtained solution. The corresponding Lagrangian can be
defined as
\begin{equation}
\overline{L}(\boldsymbol{\epsilon}^*,\boldsymbol{\lambda}^*,\boldsymbol{\rho}^*,\boldsymbol{\sigma}^{\boldsymbol{\epsilon}^*},\boldsymbol{\sigma}^{\boldsymbol{\lambda}^*},\boldsymbol{\sigma}^{\boldsymbol{\rho}^*})
=\mathrm{f}^d(\boldsymbol{\epsilon}^*,\boldsymbol{\lambda}^*,\boldsymbol{\rho}^*)
-
({\boldsymbol{\epsilon}^*}^T-\mathbf{c}^T)\boldsymbol{\sigma}^{\boldsymbol{\epsilon}^*}-\boldsymbol{\sigma}^{\boldsymbol{\lambda}^*}({\boldsymbol{\lambda}^*}^T-\mathrm{d}^T)
\end{equation}
where
$(\boldsymbol{\sigma}^{\boldsymbol{\epsilon}^*},\boldsymbol{\sigma}^{\boldsymbol{\lambda}^*})\in(\mathbb{R}^N,\mathbb{R}^K)$
are the Lagrange multipliers associated to the constraints
(\ref{New-Conc-Prob}a-\ref{New-Conc-Prob}b) respectively. The
corresponding KKT conditions are:
\begin{eqnarray}
\frac{\partial \overline{L}}{\partial\boldsymbol{\epsilon}^*}&=& 0
\quad \Rightarrow
\quad\sum_{k=1}^K\sum_{j=1}^J\left\{\frac{1}{2\rho^*_{k,j}}\left(U_{k,j}
+ \rho^*_{k,j} - \lambda^*_k -\sum_{n=1}^N\epsilon^*_n
A^k_{n,j}\right)A^k_{n,j}\right\}
=1+\sigma^{\boldsymbol{\epsilon}^*}_{n}, \forall n
\label{New-KKT-1}\\
\frac{\partial \overline{L}}{\partial\boldsymbol{\lambda}^*}&=& 0
\quad \Rightarrow \quad
\sum_{j=1}^J\left\{\frac{1}{2\rho^*_{k,j}}\left(U_{k,j} +
\rho^*_{k,j} - \lambda^*_k -\sum_{n=1}^N\epsilon^*_n
A^k_{n,j}\right)\right\} =1+\sigma^{\boldsymbol{\lambda}^*}_k, \quad
\forall k
\label{New-KKT-2}\\
\frac{\partial \overline{L}}{\partial\boldsymbol{\rho}^*}&=& 0 \quad
\Rightarrow \quad\left(\frac{U_{k,j} - \lambda^*_k
-\sum_{n=1}^N\epsilon^*_n A^k_{n,j}}{\rho^*_{k,j}}\right)^2-1=0,
\quad \forall k,j \label{New-KKT-3}
\end{eqnarray}
The above equation can be solved using the sub-gradient based
algorithm in Table \ref{tab:table1}.
Moreover, in order to ensure that the solution of (\ref{New-KKT-3}) is obtained for positive $\boldsymbol{\rho}^*$, we can use the following projection in the update of $\boldsymbol{\rho}^*$:
\vspace{-6mm}
\begin{eqnarray}
&&{\boldsymbol{\rho}^*}^{(q)}\leftarrow\Pi_{\chi_{\boldsymbol{\rho}^*}}\left({\boldsymbol{\Phi}}^{(q-1)}\right) :=
\nonumber \\&&
\left\{\begin{array} {l}
{\rho_{k,j}^*}^{(q)}=\arg \min_{\rho^*_{k,j}\in\chi^*_{\sharp}}\|\Phi^{(q-1)}_{k,j}-{\rho_{k,j}^*}\|\quad \textrm{if} \quad {\Phi_{k,j}}^{(q-1)} \leq 0,\forall k,j \\
 {\rho_{k,j}^*}^{(q)}={\Phi_{k,j}}^{(q-1)}  \quad \quad \quad \quad \quad \quad \quad \quad \quad \quad \textrm{otherwise}.
\end{array}
\right.
\label{Proj-Eqn2}
\end{eqnarray}
where $\boldsymbol{\Phi}^{(q-1)}={\boldsymbol{\rho}^*}^{(q-1)}+\beta_{\boldsymbol{\rho}^*}\frac{\partial \overline{L}}{\partial\boldsymbol{\rho}^*}|_{{\boldsymbol{\rho}^*}^{(q-1)}}$ denotes the sub-gradient update, and where $\beta_{\boldsymbol{\rho}^*}$ is the step size. The above projection ensures the positivity of $\boldsymbol{\rho}^*$.
\newtheorem{Theorem5}{\textbf{Theorem}}[section]
\begin{Theorem5}
\label{Opt-Theo1}
For $\widetilde{U}_{k,j}=U_{k,j}-2\theta_{k,j}\rho^*_{k,j}$ with
$\theta_{k,j} = \{-1,0,1\}$, $\forall k,j$; there exists a primal problem $\tilde{\mathrm{f}}(\mathbf{i})$ with utilities
$\widetilde{U}_{k,j}$ replaced for $U_{k,j}$ that can be solved optimally using the algorithm in Table I. The solution ($\overline{\boldsymbol{\epsilon}}^*$,$\overline{\boldsymbol{\lambda}}^*$,$\overline{\boldsymbol{\rho}}^*$) of $\tilde{\mathrm{f}}(\mathbf{i})$ obtained using Table I is equal to the solution of the modified problem (\ref{New-Conc-Prob}).
\end{Theorem5}
\begin{proof}
See Appendix B for the proof.
\end{proof}
Moreover, we have the following result which is the corollary of Theorem 5.1.
\newtheorem{Corollary}{\textbf{Corollary}}[section]
\begin{Corollary}
\label{Corollary}
If ${\rho}_{k,j}^*<<U_{k,j},\forall k,j$; then the solution of the canonical dual problem
obtained using the sub-gradient based algorithm (Table I) provides a solution to the primal problem which is very close to the optimal solution.
\end{Corollary}
\begin{proof}
See Appendix C for the proof.
\end{proof}
\subsubsection{Analysis of the algorithm's results for $N\rightarrow\infty$}
It can be seen from the KKT equation (\ref{Eq-2}) that $U_{k,j}-\lambda^*_k -\sum_{n=1}^N\epsilon^*_n A^k_{n,j}=\rho^*_{k,j}$ when a pattern $j$ is allocated to user $k$, and $U_{k,j}-\lambda^*_k -\sum_{n=1}^N\epsilon^*_n A^k_{n,j}= - \rho^*_{k,j}$, otherwise. When the number of sub-channel is very high, there are several patterns that have nearly equal utilities $U_{k,j}$'s for each user. This is due to the fact that for high number of sub-channels, the per sub-channel utility will be very small, and since the difference of sub-channels in the patterns with high number of sub-channels will be less, their utilities will have very small difference.
Furthermore, the difference between the summation term $\sum_{n=1}^N\epsilon^*_n A^k_{n,j}$ for several patterns of user $k$ will be very small. This means that the term $U_{k,j}-\lambda^*_k -\sum_{n=1}^N\epsilon^*_n A^k_{n,j}$ for several patterns of user $k$ will be nearly equal, as $\lambda^*_k$ is the same for all the patterns of that user.

Let us assume that a pattern $j$ is allocated to user $k$. Consequently, $U_{k,j}-\lambda^*_k -\sum_{n=1}^N\epsilon^*_n A^k_{n,j}=\rho^*_{k,j}$ whereas $U_{k,j'}-\lambda^*_k -\sum_{n=1}^N\epsilon^*_n A^k_{n,j'}=-\rho^*_{k,j'}$ for all $j'\neq j$. Moreover, in view of the above discussion, the difference between $U_{k,j}-\lambda^*_k -\sum_{n=1}^N\epsilon^*_n A^k_{n,j}$ and $U_{k,j'}-\lambda^*_k -\sum_{n=1}^N\epsilon^*_n A^k_{n,j'}$ will be very small for $j$ and $j'$ with high number of sub-channels.
Thus, being equal to $U_{k,j}-\lambda^*_k -\sum_{n=1}^N\epsilon^*_n A^k_{n,j}$ and $-U_{k,j'}+\lambda^*_k +\sum_{n=1}^N\epsilon^*_n A^k_{n,j'}$ respectively, $\rho^*_{k,j}$ and $\rho^*_{k,j'}$ will both be very small compared to $U_{k,j}$ and $U_{k,j'}$ respectively.
%


\section{Numerical Analysis and Discussion}
We consider a system with 5MHz of bandwidth (i.e. LTE) divided into
$N=25$ sub-channels each having a bandwidth of 180kHz. We assume that $K=10$
uniformly distributed users are simultaneously active in the cell. The scenario
assumed is urban canyon macro which exists in dense urban areas
served by macro-cells. A frequency selective Rayleigh fading channel
is simulated where the channel gain has a small-scale Rayleigh
fading component and a large-scale path loss and shadowing
component. Path losses are calculated according to Cost-Hata Model
and shadow fading is log-normally distributed with a standard
deviation of 8dBs. The power spectral density
of noise is assumed to be -174dBm/Hz.
\subsection{Sum-utility maximization}
In simulations, we assume that the utility of the user is equal to its weighted rate where the rate is defined by Shannon's formula. In other words, the SUmax problem is equivalent to weighted-sum rate maximization. Fig. 1 plots the empirical cumulative distribution function (CDF) of sum-utility for different resource allocation algorithms. The figure illustrates the comparison of the CDF's corresponding to our proposed algorithm, both the binary-integer programming solution and the greedy algorithm proposed in \cite{I.C.Wong}, and the round robin scheme in which an equal number of consecutive sub-channels are allocated to each user in turn. The figure shows that although the greedy algorithm proposed by Wong et al. is efficient in comparison to the round robin scheme, its performance is far away from the proposed solution. Moreover, it can be seen from the figure that the results of the proposed algorithm are very close to that obtained by solving the binary-integer program which is the optimal solution.
\subsection{Joint Adaptive Modulation and Resource Allocation}
The minimum effective SNR for each modulation $\Gamma^*_m$ that ensures a target Block Error Rate at the receiver is
determined from the link-level performance curves (e.g., see\cite{3GPP.Stand3}). Fig. 2 displays the empirical CDF of
sum-cost for different algorithms when sum-cost minimization based resource allocation (RA) is performed joint with
and without adaptive modulation (AM). The figure illustrates the comparison of the CDF's corresponding to our proposed
resource allocation algorithm when joint AM and RA is performed and when RA is performed without AM, the
binary-integer programming (BIP) based solution adopted to joint AM and RA problem, and RA without AM, and the
round robin scheme in which an equal number of consecutive sub-channels are allocated to each user in turn and minimum
possible power is allocated to the users while ensuring their target data rates. The round robin scheme is used as a
baseline scheme for comparison. The RA without AM scenario considers 16QAM as the modulation scheme.
The proposed RA with fixed modulation outperforms the round robin scheme which is not unexpected.
The figure shows that the joint AM and RA
results in a significant performance improvement over the RA without AM. The performance of the proposed algorithm can be depicted from the fact that the results of the proposed algorithm nearly overlap with that of the BIP based solution both for joint AM and RA, and RA with fixed modulation scheme. We recall that the BIP based solution is the optimal solution.
\section{Conclusion}
This paper studies resource allocation and adaptive modulation in uplink
SC-FDMA systems. Sum-utility maximization, and joint adaptive modulation and sum-cost minimization problems are considered whose optimal solutions are exponentially complex in general. A polynomial-complexity optimization framework that is inspired from the recently developed canonical duality theory is derived for the solution of both the problems. Based on the resource allocation performed by the proposed framework, an adaptive modulation scheme is also proposed for the sum-utility maximization problem that determines the best constellation for each user. The optimization problems are first formulated as binary-integer programming problems and then, each binary-integer problem is transformed into a canonical dual problem in the continuous space which is a concave maximization problem. The transformation of the problem in continuous space significantly improves the performance of the system in terms of complexity. The proposed continuous space optimization framework has a polynomial complexity that is a significant improvement over exponential complexity. It is proved analytically that that under certain conditions, the solution of the canonical dual problem is identical to the solution of
the primal problem. However, if the dual solution
does not satisfy these conditions then the optimality can not be
guaranteed. Therefore, some bounds on the sub-optimality of the proposed framework when these conditions are not satisfied are also explored. The performance of the proposed canonical dual framework is assessed by comparing it with the existing algorithms in the literature. The numerical results show that the proposed framework provides integer solution to each problem which is very close to optimal.
\begin{appendix}

\subsection{Proof of Theorem 4.2}
The total complementarity function $\Xi(\textbf{i},\boldsymbol{\epsilon}^*,\boldsymbol{\lambda}^*,\boldsymbol{\rho}^*)$ is convex in $\textbf{i}$ and concave (linear) in $\boldsymbol{\epsilon}^*$, $\boldsymbol{\lambda}^*$ and $\boldsymbol{\rho}^*$. Therefore, the stationary point $(\overline{\textbf{i}},\overline{\boldsymbol{\epsilon}}^*,\overline{\boldsymbol{\lambda}}^*,\overline{\boldsymbol{\rho}}^*)$ is a saddle point of $\Xi(\textbf{i},\boldsymbol{\epsilon}^*,\boldsymbol{\lambda}^*,\boldsymbol{\rho}^*)$. Furthermore,  $\mathrm{f}^d(\boldsymbol{\epsilon}^*,\boldsymbol{\lambda}^*,\boldsymbol{\rho}^*)$ is defined by $\Xi(\overline{\textbf{i}},\boldsymbol{\epsilon}^*,\boldsymbol{\lambda}^*,\boldsymbol{\rho}^*)$ with $\overline{\textbf{i}}$ being a stationary point of $\Xi(\textbf{i},\boldsymbol{\epsilon}^*,\boldsymbol{\lambda}^*,\boldsymbol{\rho}^*)$ with respect to $\textbf{i}\in \mathcal{I}_a$. Consequently, $\mathrm{f}^d(\boldsymbol{\epsilon}^*,\boldsymbol{\lambda}^*,\boldsymbol{\rho}^*)$ is concave on $\chi^*_{\sharp}$ and the KKT point $(\overline{\boldsymbol{\epsilon}}^*,\overline{\boldsymbol{\lambda}}^*,\overline{\boldsymbol{\rho}}^*)\in\chi^*_{\sharp}$
must be its global maximizer. Thus, by the saddle mini-max theorem:
\begin{eqnarray}
\mathrm{f}^d(\overline{\boldsymbol{\epsilon}}^*,\overline{\boldsymbol{\lambda}}^*,\overline{\boldsymbol{\rho}}^*) &=& \max_{\boldsymbol{\epsilon}^*> 0}\max_{\boldsymbol{\lambda}^*> 0}\max_{\boldsymbol{\rho}^*> 0} \mathrm{f}^d(\boldsymbol{\epsilon}^*,\boldsymbol{\lambda}^*,\boldsymbol{\rho}^*)\nonumber\\
&=& \max_{\boldsymbol{\epsilon}^*> 0}\max_{\boldsymbol{\lambda}^*> 0}\max_{\boldsymbol{\rho}^*> 0} \min_{\textbf{i} \in \mathcal{I}_a }\Xi(\textbf{i},\boldsymbol{\epsilon}^*,\boldsymbol{\lambda}^*,\boldsymbol{\rho}^*)\nonumber\\
&=& \max_{\boldsymbol{\epsilon}^*> 0}\max_{\boldsymbol{\lambda}^*> 0}\max_{\boldsymbol{\rho}^*> 0} \min_{\textbf{i} \in \mathcal{I}_a}\left\{\mathrm{f}(\textbf{i}) + \boldsymbol{\epsilon}^T\boldsymbol{\epsilon}^* +\boldsymbol{\lambda}^T\boldsymbol{\lambda}^* + \boldsymbol{\rho}^T\boldsymbol{\rho}^*\right\}\nonumber
\end{eqnarray}
\begin{eqnarray}
&=& \max_{\boldsymbol{\epsilon}^*> 0}\max_{\boldsymbol{\lambda}^*> 0} \min_{\textbf{i} \in \mathcal{I}_a}\left\{\mathrm{f}(\textbf{i}) + \boldsymbol{\epsilon}^T\boldsymbol{\epsilon}^* +\boldsymbol{\lambda}^T\boldsymbol{\lambda}^* +\max_{\boldsymbol{\rho}^*> 0}\left\{\sum_{k=1}^K\sum_{j=1}^J\rho^*_{k,j}i_{k,j}(i_{k,j}-1)\right\}\right\}\nonumber\\
&=& \max_{\boldsymbol{\epsilon}^*> 0} \min_{\textbf{i} \in \mathcal{I}_a}\left\{\mathrm{f}(\textbf{i}) + \boldsymbol{\epsilon}^T\boldsymbol{\epsilon}^* +\max_{\boldsymbol{\lambda}^*> 0}\left\{\sum_{k=1}^K \lambda^*_k\left(\sum_{j=1}^Ji_{k,j}-1\right)\right\}\right\}\nonumber\\
&&\text{s.t.}\quad i_{k,j}(i_{k,j}-1)=0,\forall k,j\nonumber\\
&=&\min_{\textbf{i} \in \mathcal{I}_a}\left\{\mathrm{f}(\textbf{i}) + \max_{\boldsymbol{\epsilon}^*> 0}\left\{\sum_{n=1}^N\epsilon^*_n\left(\sum_{k=1}^K \sum_{j=1}^Ji_{k,j}A^k_{n,j}-1\right)\right\}\right\} \nonumber\\ &&\text{s.t.}\quad i_{k,j}(i_{k,j}-1)=0, \forall k,j;\quad \sum_{j=1}^Ji_{k,j}=1,\forall k \nonumber\\
&=&\min_{\textbf{i} \in \mathcal{I}_a}\mathrm{f}(\textbf{i}) \quad \text{s.t.}\left\{i_{k,j}(i_{k,j}-1)=0,\forall k,j;\sum_{j=1}^Ji_{k,j}=1, \forall k; \sum_{k=1}^K \sum_{j=1}^Ji_{k,j}A^k_{n,j}=1,\forall n\right\} \nonumber\\
&=&\min_{\textbf{i} \in \mathcal{I}_f}\mathrm{f}(\textbf{i})
\label{proof-th4.2}
\end{eqnarray}
Note that the linear programming
\vspace{-5mm}
\begin{equation}
\max_{\boldsymbol{\rho}^*> 0}\left\{\sum_{k=1}^K\sum_{j=1}^J\rho^*_{k,j}i_{k,j}(i_{k,j}-1)\right\}
\nonumber
\end{equation}
has a finite solution in the open domain $\chi^*_{\sharp}$ if and only if $i_{k,j}(i_{k,j}-1)=0,\forall k,j$. By a similar argument, the solution of $\max_{\boldsymbol{\lambda}^*> 0}\left\{\sum_{k=1}^K \lambda^*_k\left(\sum_{j=1}^Ji_{k,j}-1\right)\right\}$ and $\max_{\boldsymbol{\epsilon}^*> 0}\left\{\sum_{n=1}^N\epsilon^*_n\left(\sum_{k=1}^K \sum_{j=1}^Ji_{k,j}A^k_{n,j}-1\right)\right\}$ leads to the last equation (\ref{proof-th4.2}). This shows that the KKT point $(\overline{\boldsymbol{\epsilon}}^*,\overline{\boldsymbol{\lambda}}^*,\overline{\boldsymbol{\rho}}^*)$ maximizes $\mathrm{f}^d(\boldsymbol{\epsilon}^*,\boldsymbol{\lambda}^*,\boldsymbol{\rho}^*)$ over $\chi^*_{\sharp}$ if and only if $\overline{\textbf{i}}$ is the global minimizer of $\mathrm{f}(\textbf{i})$ over $\mathcal{I}_f$. This completes the proof.

\vspace{-3mm}
\subsection{Proof of Theorem 5.1}

Using sub-gradient method with projection defined by (\ref{Proj-Eqn2}) ensures the positive solution of KKT equation (\ref{New-KKT-3}) which implies that the corresponding $i_{k,j}$ is binary integer. However, respecting the positivity constraint on $\boldsymbol{\lambda}^*$, equation (\ref{New-KKT-2}) can not ensure that a single sub-channel pattern is allocated to each user but $1+\sigma^{\boldsymbol{\lambda}^*}_k$ number of patterns will be allocated to each user $k$. Similarly, ensuring that $\boldsymbol{\epsilon}^*>0$, equation (\ref{New-KKT-1}) means that a sub-channel can be allocated to more than one users.

In the following, we discuss that we can find another approximate problem for which the above KKT equations not only provide binary integer solution but also ensure that a user will be assigned with a single sub-channel pattern and a sub-channel will be allocated to a single user. To this end, we proceed as follows. The KKT equation (\ref{New-KKT-3}) can be written as
\vspace{-3mm}
\begin{equation}
U_{k,j} - \lambda^*_k -\sum_{n=1}^N\epsilon^*_n A^k_{n,j}= \pm \rho^*_{k,j}, \quad \forall k,j
\label{New-KKT-32}
\end{equation}
We introduce $KJ$ new variables $\theta_{k,j}$'s defined as follows
\vspace{-3mm}
\begin{eqnarray}\label{}
   \theta_{k,j}= \left \{ \begin{array}{lll}
     \{1,0\} & \textrm{if $U_{k,j} - \lambda^*_k -\sum_{n=1}^N\epsilon^*_n A^k_{n,j}= - \rho^*_{k,j}$}    \\
     \{-1,0\} &  \textrm{if $U_{k,j} - \lambda^*_k -\sum_{n=1}^N\epsilon^*_n A^k_{n,j}= + \rho^*_{k,j}$}
     \end{array} \right.
\end{eqnarray}
From the above definition of $\theta_{k,j}$, equations (\ref{New-KKT-32}) can be written as
\vspace{-3mm}
\begin{equation}
U_{k,j} -2\theta_{k,j}\rho^*_{k,j}- \lambda^*_k -\sum_{n=1}^N\epsilon^*_n A^k_{n,j}= \pm \rho^*_{k,j}, \quad \forall k,j
\label{New-KKT-33}
\end{equation}
Let $\widetilde{U}_{k,j}=U_{k,j} - 2\theta_{k,j}\rho^*_{k,j}$, then the above equations take the form:
\vspace{-3mm}
\begin{equation}
\widetilde{U}_{k,j}- \lambda^*_k -\sum_{n=1}^N\epsilon^*_n A^k_{n,j}= \pm \rho^*_{k,j}, \quad \forall k,j
\label{New-KKT-33}
\end{equation}
Although the utilities are changed from $U_{k,j}$ to $\widetilde{U}_{k,j}=U_{k,j} - 2\theta_{k,j}\rho^*_{k,j}$, the solution of the above equations provide integer solution to $i_{k,j}$'s. We now apply this change in utilities to the equations (\ref{New-KKT-1}-\ref{New-KKT-2}).
The KKT equations (\ref{New-KKT-2}) can be written as
\vspace{-2mm}
\begin{eqnarray}
\sum_{j=1}^J\left\{\frac{1}{2\rho^*_{k,j}}\left(U_{k,j} - 2\theta_{k,j}\rho^*_{k,j} + 2\theta_{k,j}\rho^*_{k,j} + \rho^*_{k,j} - \lambda^*_k -\sum_{n=1}^N\epsilon^*_n A^k_{n,j}\right)\right\} =1+\sigma^{\boldsymbol{\lambda}^*}_k, \quad \forall k
\label{New-KKT-22}
\end{eqnarray}
Replacing $\widetilde{U}_{k,j}$ for $U_{k,j} - 2\theta_{k,j}\rho^*_{k,j}$, the above equations become:
\vspace{-2mm}
\begin{eqnarray}
\sum_{j=1}^J\left\{\frac{1}{2\rho^*_{k,j}}\left(\widetilde{U}_{k,j}+ \rho^*_{k,j} - \lambda^*_k -\sum_{n=1}^N\epsilon^*_n A^k_{n,j}\right)\right\}+\sum_{j=1}^J\theta_{k,j} =1+\sigma^{\boldsymbol{\lambda}^*}_k, \quad \forall k
\label{New-KKT-22}
\end{eqnarray}
If there exist $\theta_{k,j}$'s such that
$\sum_{j=1}^J\theta_{k,j}=\sigma^{\boldsymbol{\lambda}^*}_k$, then
we have
\vspace{-2mm}
\begin{eqnarray}
\sum_{j=1}^J\left\{\frac{1}{2\rho^*_{k,j}}\left(\widetilde{U}_{k,j}+ \rho^*_{k,j} - \lambda^*_k -\sum_{n=1}^N\epsilon^*_n A^k_{n,j}\right)\right\} =1, \quad \forall k
\label{New-KKT-23}
\end{eqnarray}
This implies that there exist another problem with a different set of utilities for which the above solution ensures that a single pattern will be allocated to each user. By using a similar procedure for the KKT equations (\ref{New-KKT-1}), we get
\vspace{-2mm}
\begin{equation}
\sum_{k=1}^K\sum_{j=1}^J\left\{\frac{1}{2\rho^*_{k,j}}\left(\widetilde{U}_{k,j} + \rho^*_{k,j} - \lambda^*_k -\sum_{n=1}^N\epsilon^*_n A^k_{n,j}\right)A^k_{n,j}\right\} =1, \forall n
\label{New-KKT-12}
\end{equation}
which enures that a sub-channel will be allocated to a single user at most when $\sum_{k=1}^K\sum_{j=1}^J\theta_{k,j}\mathbf{A}^k_{n,j}=\sigma^{\boldsymbol{\epsilon}^*}_n$, and the utilities are changed from $U_{k,j}$ to $\widetilde{U}_{k,j}=U_{k,j} - 2\theta_{k,j}\rho^*_{k,j}$.

The above analysis shows that the solution of the problem
$\mathcal{P}1$, namely
$(\overline{\boldsymbol{\epsilon}}^*,\overline{\boldsymbol{\lambda}}^*,\overline{\boldsymbol{\rho}}^*)$
that lies in the positive cone, is the solution of the above KKT
equations (\ref{New-KKT-33},\ref{New-KKT-23},\ref{New-KKT-12}).
Moreover, the KTT equations
(\ref{New-KKT-33},\ref{New-KKT-23},\ref{New-KKT-12}) give the
stationary point of a slightly modified problem
$\tilde{\mathrm{f}}^d(\boldsymbol{\epsilon}^*,\boldsymbol{\lambda}^*,\boldsymbol{\rho}^*)$
which is the canonical dual of a slightly modified primal problem with
utilities $\widetilde{U}_{k,j}=U_{k,j} - 2\theta_{k,j}\rho^*_{k,j}$.
Since the solution
$(\overline{\boldsymbol{\epsilon}}^*,\overline{\boldsymbol{\lambda}}^*,\overline{\boldsymbol{\rho}}^*)$
is positive, according to Theorems 4.1 and 4.2, the proposed
sub-gradient based solution proposed in Table \ref{tab:table1}
optimally solves a corresponding primal problem with utilities
$\widetilde{U}_{k,j}$'s and an objective function
$\tilde{\mathrm{f}}(\textbf{i})$. Note also that the canonical dual
$\tilde{\mathrm{f}}^d(\boldsymbol{\epsilon}^*,\boldsymbol{\lambda}^*,\boldsymbol{\rho}^*)$
is concave (since the KKT solution is in the positive cone).
However, how far the solution of the modified problem will be from
that of the primal problem (\ref{general}) depends upon the values
of $\rho_{k,j}^*$'s.
\vspace{-3mm}
\subsection{Proof of Corollary 5.1}
If $\rho_{k,j}^*<<U_{k,j},\forall k,j$, then
$\widetilde{U}_{k,j}\approx U_{k,j},\forall k,j$,
$\tilde{\mathrm{f}}(\textbf{i})\approx\mathrm{f}(\textbf{i})$, and
\vspace{-3mm}
\begin{equation}
\max_{(\boldsymbol{\epsilon}^*,\boldsymbol{\lambda}^*,\boldsymbol{\rho}^*)} \tilde{\mathrm{f}}^d \approx \max_{(\boldsymbol{\epsilon}^*,\boldsymbol{\lambda}^*,\boldsymbol{\rho}^*)} \mathrm{f}^d
\end{equation}
For $\rho_{k,j}^*<<U_{k,j},\forall k,j$, the solution of the
equations (\ref{New-KKT-33},\ref{New-KKT-23},\ref{New-KKT-12}) is
very close to that of equations (\ref{Eq-1},\ref{Eq-2}, \ref{Eq-3}).
Furthermore, the solution of
(\ref{New-KKT-33},\ref{New-KKT-23},\ref{New-KKT-12}) is the optimal
solution of the corresponding primal problem with utilities
$\widetilde{U}_{k,j}$ (which is very close to the optimal solution
of the primal problem with utilities $U_{k,j}$). Consequently, the
dual canonical problem obtained using the sub-gradient based
algorithm (Table \ref{tab:table1}) will provide solution to the
primal problem which is very close to the optimal solution. This
completes the proof.
\end{appendix}

\vspace{-5mm}
\begin{figure}[ht]
\centering
\includegraphics[width=.65\textwidth,height=.35\textheight]{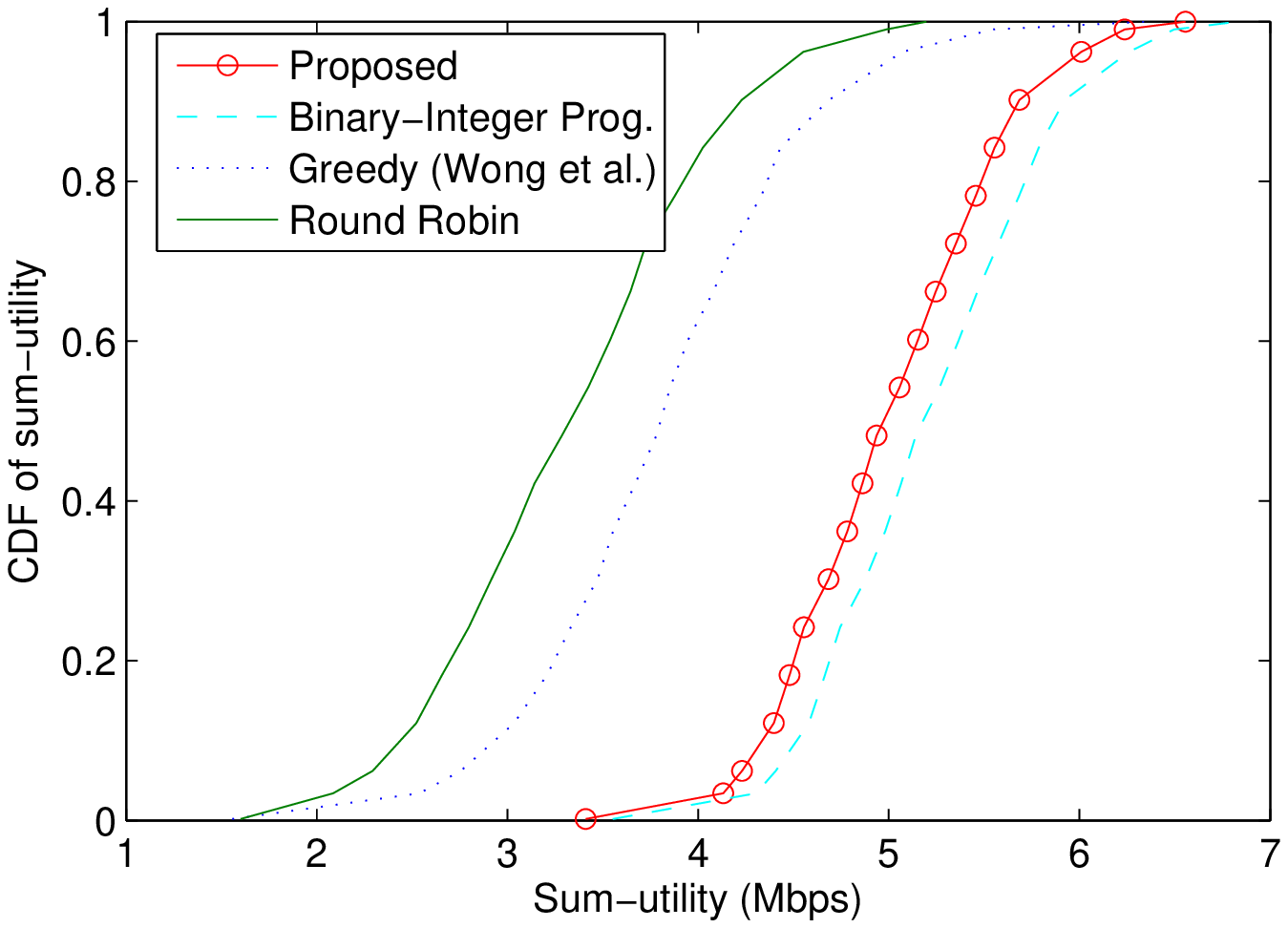}
\vspace{-8mm} \caption{Empirical CDF of sum-utility} \label{fig:Fig.
1}
\end{figure}
\vspace{-8mm}
\begin{figure}[ht]
\centering
\includegraphics[width=.65\textwidth,height=.35\textheight]{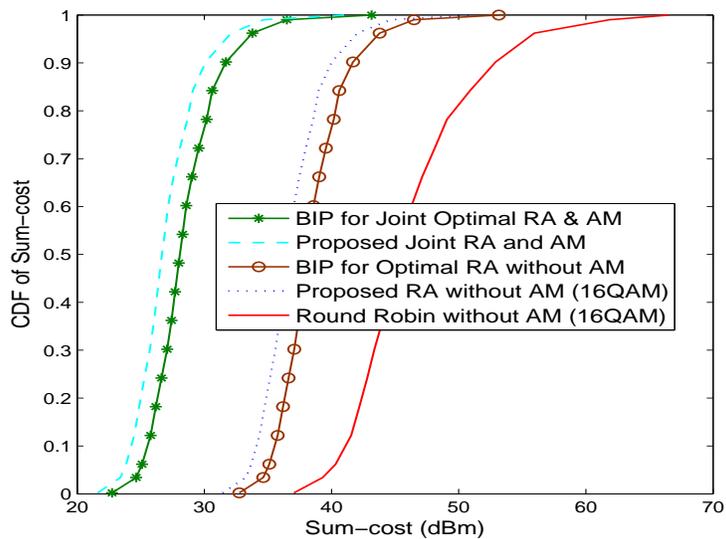}
\vspace{-8mm} \caption{Empirical CDF of sum-cost} \label{fig:Fig. 2}
\end{figure}
\end{document}